%% file: paper_draft.tex
\documentclass[12pt]{article}
\usepackage[left=1in, top=1in, right=1in, bottom=1in]{geometry}

\usepackage{setspace}

\usepackage{csquotes}
\usepackage{microtype}
\usepackage{graphicx}
\usepackage{subfigure}
\usepackage{booktabs} 
\usepackage{amsmath, amsthm, amssymb}
\usepackage{mathtools}
\usepackage{thm-restate}
\usepackage[authoryear, round]{natbib}
\usepackage{color}
\usepackage{lscape}
\usepackage{soul}
\usepackage{bm, bbm}
\usepackage{comment}
\usepackage{enumerate}
\usepackage{tabto}

\usepackage{hyperref}
\hypersetup{
    colorlinks=false,
    linkcolor=blue,
    filecolor=magenta,      
    urlcolor=cyan,
}

\include{notation}

\allowdisplaybreaks

\title{From Monopoly to Competition:\\ Optimal Contests Prevail}

\author{Xiaotie Deng\thanks{Center on Frontiers of Computing Studies, Department of Computer Science, Peking University} \\ \tt{xiaotie@pku.edu.cn} \and Yotam Gafni\thanks{Technion - Israel Institute of Technology} \\ \tt{yotam.gafni@campus.technion.ac.il} \and Ron Lavi\footnotemark[2] \thanks{University of Bath, UK} \\ \tt{ronlavi@ie.technion.ac.il} \and Tao Lin\thanks{School of Engineering and Applied Sciences, Harvard University} \\ \tt{tlin@g.harvard.edu} \and Hongyi Ling\footnotemark[1] \\ \tt{hy\_ling@pku.edu.cn}  
}
\date{}

\begin{document}
\onehalfspacing

\maketitle

\begin{abstract}
We study competition among contests in a general model that allows for an arbitrary and heterogeneous space of contest design, where the goal of the contest designers is to maximize the contestants' sum of efforts. Our main result shows that optimal contests in the monopolistic setting (i.e., those that maximize the sum of efforts in a model with a single contest) form an equilibrium in the model with competition among contests. Under a very natural assumption these contests are in fact dominant, and the equilibria that they form are unique. Moreover, equilibria with the optimal contests are Pareto-optimal even in cases where other equilibria emerge. In many natural cases, they also maximize the social welfare.

\end{abstract}

\section{Introduction}

Many important economic and social interactions may be viewed as contests. The designer aims to maximize her abstract utility (e.g.~workers' productivity, sales competitions, innovative ideas for new projects, useful information from contestants) by forming a contest, and contestants exert effort in hopes of winning a prize.
The design of optimal contests is by now well understood in the monopolistic (single-contest) setting. In particular, in many cases, a winner-takes-all contest is optimal in terms of maximizing either the sum of contestants' efforts or the single maximal effort (e.g.~\citealp{barut_symmetric_1998, kalra2001designing,moldovanu2001optimal,terwiesch2008innovation,chawla2019optimal}).


%

While most of the existing literature on contest design focuses on a monopolistic contest with an exogenously given set of participants, in reality, many times, there are multiple contests on a market and these contests must compete to attract participants, which induces a participation vs.~effort trade-off. Although the optimal contest in the single-contest setting induces maximal effort exertion {\em after} contestants choose to participate in their contests, contestants might at the same time be discouraged from choosing the more demanding contests. 
To attract participants, it seems helpful to design lucrative and easy contests that leave a large fraction of the total surplus to contestants. 
Thus, these two aspects appear to be contradicting. 

%

Previous literature has already started to acknowledge this issue with few models that formally studied it (e.g., \cite{azmat_competition_2009, stouras_prizes_2020}). In particular,~\cite{azmat_competition_2009} conclude that, despite the competitive environment, contest designers should still choose effort-maximizing contests since the effort aspect dominates the participation aspect. However, it is not clear how robust this conclusion really is, since these previous papers analyze models that are restricted in two main aspects. First, they assume that all contests have the same total prize to offer.\footnote{Previous analysis seems to significantly rely on this assumption. It also formally assumes exactly two competing contests, although this aspect may be more easily generalized.}
Second, and perhaps even more important, they restrict the choice of a contest and assume that designers choose a multiple-prize contest where contestants' winning probabilities for each prize are determined by a Tullock success function that is exgonenous and identical for all contest designers. A main appeal of a Tullock contest as a model to winner-determination in real-life contests is that it captures the fact that the efforts $e_1,...,e_k$ of the $k$ contestants in a certain contest cannot be fully observed by the contest designer. The winning probability may be following a Tullock contest success function $e_i^\tau/(\sum_{j=1}^k e_j^\tau)$ using some parameter $\tau$ to capture partial-observability by letting contestants with higher efforts win with higher probabilities (as $\tau$ increases, effort observability is better). With such a motivation, it seems that a contest designer always has the strategic option to artificially reduce her ability to observe effort, e.g., using a Tullock contest success function with some parameter $\tau' < \tau$. Combining various prize structures with a limited choice of the parameter $\tau$ is a natural way to expand the set of possible contests to consider. This is not captured in previous models. Alternatively, effort-observability can be endogenously determined (e.g., online contests could use better technology to increase effort observability). In this case, contest designers have the freedom to increase $\tau$ where in the limit as $\tau \rightarrow \infty$ we have an all-pay auction where the contestant with the highest effort wins with certainty.
Contest success functions may depend on the number of contestants in other complex ways and many additional examples of natural classes of contests exist 
\citep{corchon_theory_2007}.
Naturally, as the strategic flexibility of contest designers increases, existing outcomes may no longer be in equilibrium, and even if they do remain in equilibrium, additional more attractive equilibria might emerge. 

This paper provides a more general framework and analysis of competition among multiple contests. To be consistent with previous literature our starting point is the model of \cite{azmat_competition_2009}
which we generalize in order to capture the two main aspects discussed in the previous paragraph: a general contest design space and asymmetric contest designers.
In a high-level, the model is composed of three phases. In the first phase, contest designers choose their contests (and commit to them) from a class of contests available to them which could be any arbitrary class of contests. In the second step, after seeing the contests chosen by designers, each contestant chooses (possibly in a random way) one contest to participate in. Finally, in each contest, contestants invest effort by playing a symmetric Nash equilibrium (which previous literature has shown to exist, see details in Section~\ref{sec:model}). Designers aim to maximize the sum of efforts exerted in their own contests and contestants aim to maximize the reward they receive minus their effort.


The bottom line of our analysis is that optimal (effort-maximizing) contests in the monopolistic setting are still in equilibrium even when significantly increasing the strategic flexibility of contest designers. In other words, effort indeed dominates participation in the aforementioned trade-off for the competing designers. In fact, these contests remain the unique equilibrium in many interesting and natural cases (although, as we show, not always). Moreover, even when additional contests emerge as equilibria, choosing effort-maximizing contests is a Pareto-optimal outcome for contest designers maintaining the attractiveness (for the contest designers) of these types of contests. These conclusions hold regardless of the number of designers, the total rewards they have, and the classes of contests they can choose from.

Technically, we identify two properties that characterize the class of effort-maximizing contests and show that if every contest designer chooses a contest that satisfies these two properties then we are at an equilibrium which is Pareto optimal for the designers. The first property, which we term Monotonically Decreasing Utility (MDU) simply says that a contestant's symmetric-equilibrium utility in the single contest game decreases as the number of contestants increases. The second property, which we term Maximal Rent Dissipation (MRD), is defined with respect to the space of possible contests $\mathcal S_i$ that contest designer $i$ has. A specific contest $C_i \in \mathcal S_i$ has maximal rent dissipation if, for any other contest $C'_i \in \mathcal S_i$ that could be a possible choice for designer $i$, and for any number of contestants, $k$, the contestant's symmetric-equilibrium utility in the single contest $C_i$ when there are $k$ contestants is not larger than the contestant's equilibrium utility in the single contest $C'_i$ when there are $k$ contestants. Thus, $C_i$ minimizes the contestants utilities and therefore maximizes the utility of the contest designer, among all contests available to designer $i$. In this sense, $C_i$ is ``optimal''.

Going back to Tullock contests and prize structures, the generality of our framework yields, as corollaries to our main result, that: (1) Choosing all-pay auctions is an equilibrium for designers who can only adjust their observability of effort (i.e.~the Tullock parameter $\tau$) but must give the entire reward to the winner (i.e.~winner-takes-all); in fact we show that choosing any $\tau \geq 2$ is an equilibrium. (2) Choosing winner-takes-all contests is an equilibrium for designers who can only adjust prize structures while the parameter $\tau$ is exogenous. (3) Choosing the winner-takes-all all-pay auction is an equilibrium for the designers when they can adjust both prize structures and effort observability. 





\subsection{Additional related literature}
In this subsection we review two strands of literature on contest design. 
One strand of literature considers optimality of contests in a monopolistic (single contest) setting, in terms of revenue for the designer, agent participation, etc. A second, more recent strand, considers equilibrium outcomes
when multiple contest designers compete over agents' participation and effort. Our result is interesting in the way it ties together these domains: We show that effort-maximizing contests (those that were identified in the first literature strand)
are in an equilibrium in our general model, that belongs to and follows the second strand.
The takeaway message to a contest designer is that, in case she is interested to maximize the sum of contestants efforts, introducing competition does not change her basic goal of maximizing over effort extraction, given a fixed set of contestants. 




\subsubsection{Optimal contests in the monopolistic setting}
In the monopolistic setting, several works study optimal multiple-prize contest design with the objective of maximizing sum of efforts. Under different assumptions, most of these papers arrive at the same conclusion that the optimal contest is the winner-takes-all contest where the full prize is offered to the single contestant exerting the highest effort.  For example, in all-pay auctions where contestants' efforts are fully observable, \cite{barut_symmetric_1998, moldovanu2001optimal} show that a winner-takes-all all-pay auction is optimal assuming contestants having either linear or concave cost functions (interestingly, for convex cost functions their results vary). An exception is \cite{glazer_optimal_1988} who show that the optimal contest should offer equal prizes to all players except for the player with the lowest effort if players value the prize money by a strictly concave utility function.  We consider linear utilities and linear cost functions as in \cite{barut_symmetric_1998} so the winner-takes-all is optimal. When contestants' efforts are not fully observable and their winning probabilities for different prizes are assumed to follow a Tullock success function, the optimal prize structure is once again a winner-takes-all \citep{clark_influence_1998}.  The winner-takes-all is also optimal in a stochastic-quality model (e.g., \citealp{kalra2001designing, ales2017optimal}) where a contestant who exerts effort $e_i$ produces a submission with random quality $Q_i = e_i + Z_i$ where $Z_i$ follows some noise distribution. In this case, prizes are allocated based on the submission qualities.
If the designer must offer a single prize, then the all-pay auction is optimal among all possible contests, as it induces full rent dissipation -- contestants' utilities are reduced to zero and their sum of efforts is maximized \citep{baye1996apa}. 

A line of the literature including \cite{segev_schweinzer_2012,baye1996apa,alcalde2010rent,ewerhart_2017_mdu} studies symmetric equilibria and rent dissipation in optimal contests. Our formal results, which analyze contests that admit symmetric equilibria and certain dissipation properties, have concrete applications thanks to existence and characterization results provided in these papers. 


\subsubsection{Competition among contests}

Previous works (e.g., \citealp{stouras_prizes_2020, azmat_competition_2009}) on contest competition suggest that either participation or effort may dominate in the aforementioned trade-off under different assumptions.  For example, \cite{stouras_prizes_2020} show that, for designers who wish to maximize the highest quality, participation outweighs effort and hence the designers will set multiple prizes when the quality of submission is sensitive to their effort. When the quality of submission is not sensitive to effort, they show that the effort aspect is dominating and hence the designers will offer a single prize which induces the maximum effort exertion in the single-contest setting.  On the other hand, \cite{azmat_competition_2009} suggest that, for designers maximizing the sum of efforts, the effort aspect is always dominating, and hence a single prize should be offered. As mentioned before, we generalize their work in several aspects.  

\cite{crowdsourcingSystemsAnalysis} consider an incomplete information competition among contests. They focus on participation issues rather than on the strategic choices of contest designers, by assuming that all contests are all-pay. They explicitly characterize the relationship between contestants' participation behavior and contests' rewards, and
find that rewards yield logarithmically diminishing returns with respect to participation levels. 
\cite{korpeouglu2017parallel} consider an incomplete information contest model where  contestants can participate in multiple contests, and contest designers use winner-takes-all contests while strategically choosing rewards to maximize the maximal submission quality minus reward. They show that, in several cases, contest designers benefit from contestants' participation in multiple contests.

\section{Competition Among Contests: Model and Preliminaries}
\label{sec:model}

\subsection{A single-contest game}


A contest designer designs a contest among several contestants in order to maximize the sum of efforts exerted by the contestants in return for some reward to be divided among them according to some winning rule determined by the designer.

Formally,
a contest $C$ is composed of a reward $R$ and a family of contest success functions $\vec{f}^k:\mathbb{R}_{\geq0}^k\to[0,1]^k$ for each number of contestants $k>0$. Contestants exert efforts $(e_1, \ldots, e_k) \in \mathbb R_{\ge 0}^k$ to compete for the reward. Each contestant $i$ receives a fraction $f^k_i(e_1,\ldots,e_k)$ of the reward, where $f^k_i(e_1,\ldots,e_k)$ is the $i$-th coordinate of the vector $\vec{f}^k(e_1,\ldots,e_k)$. 
In a stochastic quality model (e.g., the additive-noise model\footnote{In an additive-noise contest model, a contestant who exerts effort $e_i$ produces a submission with random quality $Q_i = e_i + Z_i$ where $Z_i$ follows some noise distribution.  The contest designer observes $Q_i$ but not $e_i$ and allocate rewards to contestants based on their $Q_i$'s. }), $f_i^k(e_1, \ldots, e_k)$ is the expected fraction of reward received by contestant $i$.
We allow general functions $f_i^k(\cdot)$ and only require that $\sum_{i=1}^kf^k_i(e_1,\ldots,e_k)\leq1$. The utility of a contestant is the reward she gets minus the effort she exerts: $f^k_i(e_1,\ldots,e_k) R - e_i$. The utility of the contest designer is the sum of efforts $\sum\limits_{i=1}^ke_i$. When $k=0$ the designer's utility is 0.



\begin{definition}
A contest is \emph{anonymous} if its contest success functions $\vec{f}^k:\mathbb{R}_{\geq0}^k\to[0,1]^k$ satisfy, for any $k>0$, for any $(e_1,\ldots,e_k) \in \mathbb{R}_{\geq0}$ and any permutation $\pi$ of $(1,\ldots,k)$,
\begin{equation*}
    \vec{f}^k\left(e_{\pi(1)},\ldots,e_{\pi(k)}\right)=\left(f^k_{\pi(1)}(e_1,\ldots,e_k),\ldots,f^k_{\pi(k)}(e_1,\ldots,e_k)\right).\footnote{This is the same definition as in \cite{alcalde2010rent}; it is equivalent to requiring that $f^k_i(e_1,\ldots,e_k)=
    f^k_{\pi(i)}(\tilde{e}_1,\ldots,\tilde{e}_k)$ where $\tilde e_{\pi(j)} = e_j$ for all $j$.}
\end{equation*}
\end{definition}

\begin{definition}
A contest \emph{fully allocates the reward} if its contest success functions $\vec{f}^k:\mathbb{R}_{\geq0}^k\to[0,1]^k$ satisfy, for any $k>0$ and any $(e_1,\ldots,e_k) \in \mathbb{R}_{\geq0}^k$, $\sum_{i=1}^kf^k_i(e_1,\ldots,e_k)=1$.
\end{definition}

\begin{example}
A Tullock contest (or, more accurately, a single-prize Tullock contest) parameterized by  $\tau\in[0,+\infty]$ has the following contest success function: 
\[ 
f_i^k(e_1,\ldots,e_k)=\begin{cases}
\frac{e_i^\tau}{\sum\limits_{j=1}^ke_j^\tau} & \text{if $e_j>0$ for some $j\in\{1, \ldots, k\}$} \\
\frac{1}{k} & \text{otherwise}
\end{cases} 
\]
When $\tau=+\infty$, the contest becomes an ``All Pay Auction (APA)'' where the contestant with the highest effort wins with certainty (to maintain anonymity, if several contestants exert the highest effort, they all win with equal probability).
A Tullock contest is anonymous and it fully allocates the reward.


\end{example}

\begin{definition}
Denote by $\mathcal{C}_R$ the set of all contests with reward $R$ that are anonymous, fully allocate the reward and have a symmetric Nash equilibrium among $k$ contestants for all $k>0$.
\end{definition}

For example, \cite{alcalde2010rent} and \cite{baye1996apa} show that Tullock contests with parameters $\tau\in[0, \infty)$ and $\tau=\infty$ admit a symmetric Nash equilibrium; thus, $\mathcal{C}_R$ contains all Tullock contests with reward $R$. Other examples of contests that admit a symmetric Nash equilibria are given in e.g., the seminal works of \cite{hirshleifer1989conflict,Nti1997ComparativeSO}, a survey by \cite{corchon_theory_2007}, as well as later works such as  \cite{amegashie2012nested}. 


We assume throughout the paper that all contestants in the same contest will play a symmetric Nash equilibrium of that contest. Formally, for every contest $C \in \mathcal{C}_R$ we fix a (mixed strategy) symmetric Nash equilibrium, i.e., $e_1, \ldots, e_i, \ldots, e_k$ are i.i.d.~random variables that follow a distribution $F$ defined by a mixed strategy Nash equilibrium. Since $C \in \mathcal C_R$ is anonymous, in the symmetric Nash equilibrium all contestants get an equal expected fraction of the reward and hence their expected utilities are identical. We denote their identical expected utility by
$\cutil_C(k) = \Ex[e_1, \ldots, e_k \sim F]{  f_i^k(e_1, \ldots, e_k) R - e_i}$. 
Moreover, since $C\in \mathcal C_R$ fully allocates the reward, we must have $ \Ex[e_1, \ldots, e_k \sim F]{f_i^k(e_1, \ldots, e_k)} = \frac{1}{k}$ and hence 
\begin{equation}
\label{eq:cutil-R-k}
\cutil_C(k) = \frac{R}{k} -  \Ex[e_i\sim F]{e_i}.
\end{equation}
\noindent
We note that when $k=1$, $\cutil_C(1)=R$, because the single contestant will not exert any effort. We also have $\cutil_C(k)\geq0$ for any $k>0$ since a contestant can always choose to exert zero effort and guarantee non-negative utility.  Moreover, since $e_i\ge 0$, we always have $\cutil_C(k)\leq \frac{R}{k}$. 
We can use $\cutil_C(k)$ to express the utility of a contest designer in a contest $C\in \mathcal C_R$ with $k\ge 1$ contestants by rearranging~\eqref{eq:cutil-R-k}: 
\begin{equation}\label{eq:sum-R-gamma}
    \Ex[e_1, \ldots, e_k \sim F]{\sum\limits_{i=1}^k e_i} =  k \Ex[e_i \sim F]{e_i} = R - k\cutil_C(k). 
\end{equation}
Note that we assume that the utility of a contest designer is the expected sum of efforts, even if this is non-observable. This fits settings like workplace contests that aim to improve workers' productivity. More generally, in the additive noise model, expected sum of qualities is equal to expected sum of efforts since the expected noise is usually assumed to be zero.

\subsection{A contest competition game}

In this paper we study a game where multiple contest designers compete by choosing their contest success functions. Contestants observe the different contests and choose in which one to participate.

\begin{definition}
A complete-information \emph{contest competition game} is denoted by $CCG(m, n, $ $(R_i)_{i=1}^m,$ $(\mathcal S_i)_{i=1}^m)$, where $m\ge 2$ is the number of contest designers, $n\ge 1$ is the number of contestants, $R_i>0$ is the reward of contest $i$, and $\mathcal S_i \subseteq \mathcal C_{R_i}$. The game has two phases: 
\begin{enumerate}
    \item \textbf{Designers choose contests}.  Each designer $i$ chooses a contest $C_i \in \mathcal S_i$ simultaneously. Contestants observe the chosen contests $(C_1, \ldots, C_m)$.
    
    \item \textbf{Contestants play a normal-form game of choosing in which contest to participate}.
    A pure strategy of each contestant in this game is to choose one contest.
    Importantly, contestants may play a mixed strategy, meaning that each contestant $\ell$ participates in each contest $C_i\ (i=1,...,m)$ with some probability $p_{\ell i}$, $\sum\limits_{i=1}^m p_{\ell i} = 1$. We denote the vector of probabilities chosen by contestant $\ell$ by $\vec{p}_{\ell}=(p_{\ell 1},...,p_{\ell m})$.

\end{enumerate}


\noindent
After Nature assigns contestants to contests, utilities are as follows. If there are $k\ge 1$ contestants participating in contest $C_i$, then each of these contestants 
gains utility $\cutil_{C_i}(k)$ and contest designer $i$ gains utility $R_i - k\cutil_{C_i}(k)$. If $k=0$ then the utility of the contest designer is $0$.
\end{definition}

The first important element of our model is the space $\mathcal S_i$ of all possible contests a designer can (strategically) choose. For example, this could be the space of all Tullock contests, i.e., the parameter $\tau$ becomes a strategic choice. Some of the previous literature views $\tau$ as an exogenous parameter representing how accurately the designer is able to observe the ranking of efforts performed by contestants. Even so, it seems plausible that the designer chooses an ``ignorance is bliss'' approach where she lowers the $\tau$ value (thus, observes efforts' ranking less accurately) in order to encourage participation.
Alternatively, another example for $\mathcal S_i$ could be the space of all prize structures where the ranking is determined according to a specific Tullock contest with a fixed exogenous $\tau$, as in~\cite{azmat_competition_2009}.\footnote{Since we assume contests that fully allocate the reward, each $f^k$ can be derived from prize structures of at most $k$ prizes.~\cite{azmat_competition_2009} allow for any number of prizes, although having more prizes than contestants is not beneficial for the contest designer.}

%
We remark that this model implicitly assumes that when a contestant decides on the level of effort to exert in the contest she participates in, 
she knows the total number of contestants $k$ in the same contest. In practice, contestants can observe the number of participants when they are in physical contests (like sport contests) or when the contest designer chooses to reveal this information. Moreover, \cite{myerson_population_2006} show that contest designers have an incentive to do so because the expected aggregate effort in a contest with a commonly known number of participants is in general higher than that in a contest where the contestants do not see the number of participants. 


In the second phase each contestant has a finite number $m$ of possible actions and the game is symmetric, hence there must exist at least one symmetric (mixed strategy) Nash equilibrium \citep{nash_non-cooperative_1951}. We will assume in all our results that the contestants play this symmetric equilibrium, i.e., we will only consider equilibria in which the probability vector of every contestant is the same ($\vec{p}_1=\vec{p}_2=\cdots=\vec{p}_n$). Example~\ref{ex:asymmetric} discusses the case where the contestants choose an asymmetric equilibrium.
Formally, we denote by $\bm{p}(C_1, \ldots, C_m)\in\reals^m$ the probability vector chosen by the contestants at their symmetric equilibrium when the designers choose contests $(C_1, \ldots, C_m)$ in the first phase
.\footnote{If there are multiple symmetric equilibria, we allow $\bm{p}(C_1, \ldots, C_m)$ to be any one of those.  All our conclusions hold regardless of which symmetric equilibrium the contestants play. In addition, we show in Lemma~\ref{lem:MDU-contestant-unique} that the symmetric equilibrium is unique if a certain condition (that is satisfied, e.g., by all Tullock contests) holds. }

Given that designers choose contests $\bm C = (C_1, \ldots, C_m)$ in the first phase of the game and contestants participate in contests with probabilities $(p_1, \ldots, p_m) = \bm{p}(\bm C)$ in equilibrium, the contestants' utility is as follows. 
For a contestant who participates in $C_i$, the number of contestants among the other $n-1$ contestants who also participate in $C_i$ follows the binomial distribution $\bin(n-1, p_i)$.  Therefore, the expected utility of a contestant participating in $C_i$, denoted by
$\beta(C_i, p_i)$, equals
\begin{equation}\label{eq:contestant-utility-beta}
 \beta(C_i, p_i) = \Ex[k\sim \bin(n-1, p_i)]{\cutil_{C_i}(k+1)} = \sum\limits_{k=0}^{n-1} \binom{n-1}{k} p_i^k (1-p_i)^{n-1-k} \cutil_{C_i}(k+1).
\end{equation}

Denote the set of indices of contests in which contestants participate with positive probability (i.e., the support of $\bm p(\bm C)$) by
\begin{equation}
\Supp(\bm C)=\{i:p_i(\bm C)>0\}.
\end{equation}

\begin{claim}[Equilibrium condition]
\label{claim:equilibrium-condition}
Suppose that designers choose contests $\bm C = (C_1, \ldots, C_m)$ in the first phase of the game and contestants participate in contests with probabilities $(p_1, \ldots, p_m) = \bm{p}(\bm C)$ in equilibrium. Then,
\begin{itemize}
    \item If $i \in \Supp(\bm C)$, then $\beta(C_i, p_i) \ge \beta(C_j, p_j)$ for any $j=1, \ldots, m$.
    \item Thus, if $i,j \in \Supp(\bm C)$, then $\beta(C_i, p_i) = \beta(C_j, p_j)$. 
\end{itemize}
\end{claim}
\begin{proof}
$(p_1,\ldots,p_m)$ is a symmetric equilibrium, i.e., given that all other participants play $(p_1,\ldots, p_m)$, a player's best response is to play $(p_1,\ldots, p_m)$ herself. Therefore, the expected utility of choosing to participate in contest $i$ is at least as high as choosing to participate in contest $j$, as contest $i$ is assigned a positive probability $p_i > 0$. 
\end{proof}

\subsection{Equilibrium among contest designers}
We use $\bm{C} = (C_i, \bm{C}_{-i}) = (C_1, \ldots, C_m)$ to denote the contests (strategies) chosen by all designers, where  $\bm{C}_{-i}$ denotes the contests chosen by designers other than $i$.  Let $u_i(C_i, \bm C_{-i})$ be the expected utility of contest designer $i$ given that contestants use $\bm{p}(C_i, \bm C_{-i})$. Formally, by \eqref{eq:sum-R-gamma} the utility of the designer of contest $C_i$ equals $R_i - k\cutil_{C_i}(k)$ when there are $k\ge 1$ participants. Since each contestant participates in $C_i$ independently with probability $p_i = p_i(C_i, \bm C_{-i})$, the total number $k$ of participants in $C_i$ follows the binomial distribution $\bin(n, p_i)$, and hence the designer's expected utility equals 
\begin{equation}\label{eq:designer-utility-definition}
    u_i(C_i, \bm C_{-i}) = \Ex[k\sim \bin(n, p_i)]{(R_i - k\cutil_{C_i}(k)) \cdot \mathbbm{1}[k\ge 1] }. 
\end{equation} 
%

Since a contest between the designer and the multiple contestants is a constant-sum game where the overall utility of all players (i.e., the welfare) equals the total reward $R_i$ whenever there is at least one contestant, designer $i$'s expected utility $u_i(C_i, \bm C_{-i})$ can be written as the expected welfare $R_i[1 - (1-p_i)^n]$ minus the sum of contestants' expected utilities obtained from contest $i$, $n p_i \beta(C_i, p_i)$. Formally,

\begin{claim}\label{claim:designer-utility}
\begin{equation}\label{eq:designer-utility}
u_i(C_i, \bm C_{-i})=
R_i \left[ 1 - (1-p_i)^n\right] - np_i \beta(C_i, p_i).
\end{equation}
\end{claim}
\begin{proof}
\begin{align*}
    u(C_i, \bm C_{-i})&= \Ex[k\sim \bin(n, p_i)]{(R_i - k\cutil_{C_i}(k)) \cdot \mathbbm{1}[k\ge 1] }
    = \sum\limits_{k=1}^n\binom{n}{k}p_i^k(1-p_i)^{n-k}(R_i -k\cutil_{C_i}(k))\\
    &= \sum\limits_{k=1}^n\binom{n}{k}p_i^k(1-p_i)^{n-k}R_i -\sum\limits_{k=1}^n\binom{n}{k}p_i^k(1-p_i)^{n-k}k\cutil_{C_i}(k)\\
    &= R_i\left[1-(1-p_i)^n\right] - \sum\limits_{k=1}^n\binom{n}{k}p_i^k(1-p_i)^{n-k}k\cutil_{C_i}(k)\\
    &= R_i\left[1-(1-p_i)^n\right] - np_i\sum\limits_{k=1}^n \binom{n-1}{k-1} p_i^{k-1}(1-p_i)^{n-k}\cutil_{C_i}(k)\\
    &= R_i\left[ 1 - (1-p_i)^n\right] - np_i \Ex[k'\sim \bin(n-1, p_i)]{\cutil_{C_i}(k'+1)}, 
\end{align*}
which equals $R_i\left[ 1 - (1-p_i)^n\right] - np_i \beta(C_i, p_i)$ by \eqref{eq:contestant-utility-beta}. 
\end{proof}

In this paper we analyze the following solution concepts for the contest competition game:

\begin{definition}
Given some $CCG(m, n, (R_i)_{i=1}^m, (\mathcal S_i)_{i=1}^m)$,
\begin{itemize}

\item A contest $C_i \in \mathcal S_i$ is \emph{dominant} if $\forall C'_1 \in \mathcal S_1,...,C'_m \in \mathcal S_m$, $u_i(C_i, \bm{C'}_{-i}) \ge u_i(C'_i, \bm{C'}_{-i})$. 

\item A tuple of contests $(C_1, \ldots, C_m)$, where $C_i \in \mathcal S_i$ for all $i$, is a \emph{contestant-symmetric subgame-perfect equilibrium} if $u_i(C_i, \bm{C}_{-i}) \ge u_i(C_i', \bm{C}_{-i}), \forall C_i'\in \mathcal S_i, \forall i=1, \ldots, m$.

\end{itemize} 
\end{definition}

For simplicity and also for practical purposes, we do not consider the case where designers play mixed strategies (i.e., distributions over multiple contests).

\subsection{Additional important properties of contests}

Our results will use the following three properties of contests: 
\begin{definition} \label{def:MDU-MRD}
~~
\begin{itemize}
    \item A contest $C_i\in \mathcal C_{R_i}$ has \emph{monotonically decreasing utility (MDU)} if 
$\cutil_{C_i}(1) \ge \cutil_{C_i}(2) \ge \cdots \ge \cutil_{C_i}(n)$. In words, the symmetric Nash equilibrium expected utility of a contestant is decreasing as the number of contestants increases.

    \item A contest $C_i \in \mathcal S_i \subseteq \mathcal C_{R_i}$ has \emph{maximal rent dissipation (MRD)} in $\mathcal S_i$ if for any $C'_i \in \mathcal S_i$ and any $k=1, \ldots, n$, $\cutil_{C_i}(k) \leq \cutil_{C'_i}(k)$. Let $\mrd(\mathcal S_i) \subseteq \mathcal S_i$ denote the set of all contests with maximal rent dissipation in $\mathcal S_i$. In words, an MRD contest maximizes the designer's utility regardless of the number of contestants which is equivalent to minimizing the symmetric Nash equilibrium expected utility of contestants.

    \item A contest $C_i \in \mathcal{C}_{R_i}$ has \emph{full rent dissipation} if $\cutil_{C_i}(1)=R_i$ and $\cutil_{C_i}(k)=0$ for $k=2, \ldots, n$. 
\end{itemize}
\end{definition}

\begin{restatable}[]{claim}{monotonedecreasing}
\label{claim:monotone-decreasing}
If $C_i$ has monotonically decreasing utility, then for $p < p'$, $\beta(C_i, p) > \beta(C_i, p')$. 
\end{restatable}

The proof of the claim is in Appendix~\ref{app:model_proofs}.

\begin{claim}
\label{claim:mrd-inequality}
Let $T_i\in \mrd(\mathcal S_i)$, and let $C_i \in \mathcal S_i$ be any other contest in $\mathcal S_i$. 
Then for any $p \in [0, 1]$, $\beta(T_i, p) \le \beta(C_i, p)$. 
\end{claim}
\begin{proof}
By the definition of maximal rent dissipation contest, $\cutil_{C_i}(k+1) \ge \cutil_{T_i}(k+1)$ for all $k=0, \ldots, n-1$, thus 
\[ \beta(C_i, p) = \Ex[k\sim \bin(n-1, p)]{\cutil_{C_i}(k+1)} \ge \Ex[k\sim \bin(n-1, p)]{\cutil_{T_i}(k+1)} = \beta(T_i, p). \qedhere \]
\end{proof}

Note that a full rent dissipation contest $C_i$ has monotonically decreasing utility and has maximal rent dissipation in any set $\mathcal S_i \subseteq \mathcal C_{R_i}$ that contains it.
It is known that $\wta$ has full rent dissipation~\citep{baye1996apa} 
and in fact, as a corollary of \cite{ewerhart_2017_mdu} we observe that every Tullock contest with parameter $\tau \ge 2$ has full rent dissipation. 
Thus, the class of Tullock contests contains maximal rent dissipation contests (namely, those with $\tau\ge 2$). Also, it is a class of contests that have monotonically decreasing utility: 

\begin{restatable}[Corollary of \citealp{baye1996apa, segev_schweinzer_2012, ewerhart_2017_mdu}]{lemma}{tullock}
\label{lem:tullock}
Let $C_\tau$ be a Tullock contest with reward $R$ and with parameter $\tau \in [0, +\infty]$. Then, $\cutil_{C_\tau}(k) = R(\frac{1}{k} - \frac{k-1}{k^2} \tau)$ if $\frac{k}{k-1} > \tau$ and $\cutil_{C_\tau}(k) = 0$ if $\frac{k}{k-1} \le \tau$.  For $\tau = +\infty$, $\cutil_{C_\tau}(1)=R$ and $\cutil_{C_\tau}(k)=0$ for $k\ge 2$. As corollaries, 
\begin{itemize}
    \item Every Tullock contest has monotonically decreasing utility. 
    \item Every Tullock contest with parameter $\tau \ge 2$ has full rent dissipation.
    \item If $\mathcal S$ is the set of all Tullock contests with parameter $\tau$ in some range whose maximum $\tau^{\max}$ is well defined and at most $2$, then the Tullock contest with $\tau^{\max}$ is the only contest in $\MRD(\mathcal S)$.
\end{itemize}
\end{restatable}

\noindent
A proof of this lemma is given in Appendix~\ref{app:model_proofs}.
Proposition 2 of \cite{Nti1997ComparativeSO} shows a large class of contests that generalize Tullock contests with $\tau\leq 1$ and have the MDU property. 
MDU contests have another useful property:

\begin{lemma}\label{lem:MDU-contestant-unique}
If $C_1, \ldots, C_m$ are MDU contests then the contestants' symmetric equilibrium $\bm p(C_1, \ldots, C_m)$ is unique. 
\end{lemma}
\begin{proof}
Let $\bm p=(p_1,\ldots,p_m)$ and $\bm p'=(p_1',\ldots,p_m')$ be two symmetric equilibria for contestants. If they are different, then there exist $i,j$ such that $p_i>p_i'$ and $p_j<p_j'$. Then we get the following contradiction
\begin{align*}
    & \beta(C_i,p_i') > && (p_i'< p_i,C_i\text{ has MDU, Claim~\ref{claim:monotone-decreasing}})\\
    & \beta(C_i,p_i) \ge && (p_i>0, \text{Claim~\ref{claim:equilibrium-condition}})\\
    & \beta(C_j,p_j) > && (p_j<p_j',C_j\text{ has MDU, Claim~\ref{claim:monotone-decreasing}})\\
    & \beta(C_j,p_j') \ge && (p_j'>0, \text{Claim~\ref{claim:equilibrium-condition}})\\
    & \beta(C_i,p_i'). && \hfill \qedhere
\end{align*}
\end{proof}

\section{Main Results: Equilibria in Contest Competition Games}
\label{sec:main_results}

Our first main result shows that choosing a maximal rent dissipation contest with a monotonically decreasing utility is a subgame-perfect equilibrium of the CCG game, and, moreover, that such a maximal rent dissipation contest is a dominant contest when the set of all possible contests contains only contests with monotonically decreasing utilities:
\begin{theorem}\label{thm:max-rent-dissipation-NE}
~~
\begin{enumerate}
    \item Fix any $CCG(m, n, (R_i)_{i=1}^m, (\mathcal S_i)_{i=1}^m)$ where each $\mathcal S_i \subseteq \mathcal C_{R_i}$ contains a maximal rent dissipation contest that has monotonically decreasing utility, denoted by $T_i\in\mrd(\mathcal S_i)$. Then, $(T_1, \ldots, T_m)$ is a contestant-symmetric subgame-perfect equilibrium.

    \item Moreover, if 
    each $\mathcal S_i$ only contains contests with monotonically decreasing utility,
    then $T_i$ is a dominant contest for each designer $i$. 

\end{enumerate}
\end{theorem}

The full proof, as well as most other proofs in this section are deferred to Appendix~\ref{app:main_results_proofs}. In a very high-level, the argument for why MRD contests constitute an equilibrium for the contest competition game is the following.
Consider any contest designer $i$. Suppose each of the $n$ contestants participates in designer $i$'s contest with some probability $p_i$ (assuming a symmetric participation equilibrium).  According to Claim~\ref{claim:designer-utility}, designer $i$'s expected utility equals
\[ u_i(C_i, \bm C_{-i})=
R_i \left[ 1 - (1-p_i)^n\right] - np_i \beta(C_i, p_i), \]
where we recall that $\beta(C_i, p_i)$ is each contestant's expected utility conditioning on her already participating in $C_i$. 
%
Now, suppose that contest designer $i$ switches to a contest $C'_i$ that requires less effort from the contestants (namely, leaving more utility to the contestants) and hence increases the participation probability to $p_i' = p_i + \Delta p$.  The welfare term is increased by $\Delta p \frac{\partial R_i[1 - (1-p_i)^n]}{\partial p_i}  =  n \Delta p R_i (1-p_i)^{n-1} $.  
A contestant's utility in contest $i$ will be increased (here we will use the condition that other contests $C_j$, $j\ne i$, are MDU contests, as explained in the formal proof) and suppose it is increased to $\beta(C_i', p_i')=\beta(C_i, p_i) + \Delta \beta$ ,
so the utility of each contestant is increased by
\begin{align*}
p_i'\beta(C'_i, p_i') - p_i \beta(C_i, p_i) = (p_i + \Delta p) (\beta(C_i, p_i) + \Delta \beta) - p_i \beta(C_i, p_i) & = \Delta p \beta(C_i, p_i) + p_i \Delta \beta + \Delta p \Delta \beta \\
& > \Delta p \beta(C_i, p_i) \\
& \ge \Delta p R_i (1-p_i)^{n-1},
\end{align*}
where the last inequality is because a contestant obtains utility $R_i$ when no other contestants participate in $C_i$, which happens with probability $(1-p_i)^{n-1}$. 
Thus, the overall increase of contestants' utility is greater than $n \Delta p R_i  (1-p_i)^{n-1} $, outweighing the increase of the welfare term, so the designer's utility is decreased.


Theorem~\ref{thm:max-rent-dissipation-NE} has the following implication regarding $\wta$ (or any other full rent dissipation contest):

\begin{corollary}\label{cor:full-rent-dissipation-NE}
~~
\begin{enumerate}
    \item Fix any $CCG(m, n, (R_i)_{i=1}^m, (\mathcal S_i)_{i=1}^m)$ where each $\mathcal S_i \subseteq \mathcal C_{R_i}$ contains a full rent dissipation contest (e.g., $\wta$), denoted by $F_i$. Then, $(F_1, \ldots, F_m)$ is a contestant-symmetric subgame-perfect equilibrium. 
    
    \item Let $\mathcal T_{R_i}$ be the set of all Tullock contests with reward $R_i$. 
    Then, 
    $\wta$ and any other Tullock contest with $\tau \ge 2$ is a dominant contest for every designer in $CCG(m, n, (R_i)_{i=1}^m, (\mathcal T_{R_i})_{i=1}^m)$.
    
    \item If $\mathcal S_i$ is the set of all Tullock contests with parameter $\tau_i$ in some range whose maximum $\tau^{\max}_i$ is well defined and at most $2$. Then, the Tullock contest with $\tau^{\max}_i$ is the only dominant contest for every designer in $CCG(m, n, (R_i)_{i=1}^m, (\mathcal S_i)_{i=1}^m)$. 

\end{enumerate}
\end{corollary}

Theorem~\ref{thm:max-rent-dissipation-NE} shows a specific type of 
contestant-symmetric subgame-perfect equilibria. The following example shows that if the sets $\mathcal{S}_i$ also contain contests that do not satisfy the condition of monotonically decreasing utility, other types of equilibria may exist, and the equilibria that the theorem shows is not dominant:

\begin{example}
\label{ex:non-MDU}
 Let $m=2, n=6$, $R_1 = R_2 = 1$, both $\mathcal S_1$ and $\mathcal S_2$ consist of two contests: the $\wta$ contest and a contest $C$ that gives the reward for free when the number of participants $k=5, 6$ and runs $\wta$ otherwise. We thus have $\cutil_C = (1, 0, 0, 0, 1/5, 1/6)$, which is not monotonically decreasing.  We claim that $(C, C)$ is a contestant-symmetric subgame-perfect equilibrium and that $\wta$ is not a best-response to $C$ (and therefore not a dominant contest):  When designers choose $(C, C)$, by symmetry, contestants participate in either contest with equal probability $(0.5, 0.5)$. By direct computation (e.g., using \eqref{eq:designer-utility}), the expected utility of each designers is $(0.7812, 0.7812)$. Now suppose designer $1$ switches to $\wta$. The probabilities $(p_1,p_2)=\bm p(\wta, C)$ in the contestants' symmetric mixed strategy Nash equilibrium must satisfy, according to Claim~\ref{claim:equilibrium-condition}, $\beta(\wta, p_1) = \beta(C, p_2)$ (assuming $p_1, p_2>0$). 
By numerical methods, we find that $(p_1, p_2) = (0.4061, 0.5939)$.
The expected utility of designers is then $(0.7761, 0.7323)$. Since $0.7761 < 0.7812$, designer 1 will not switch to $\wta$.  By symmetry, designer 2 will not switch to $\wta$. Hence, $(C, C)$ is an equilibrium, and $\wta$ is not a dominant contest.  
\end{example}

In this example, for every $k$, $\bm{f}^k$ uses a Tullock contest with a parameter $\tau_k$ that depends on $k$ (namely, $\tau_k = 0$ for $k=5, 6$ and $+\infty$ otherwise). An example where we use $\tau_k=1$ instead of $\tau_k=0$ for some $k$ could be constructed in a similar way.\footnote{In particular, $m=2, n=10$, $R_1 = R_2 = 1$, both $\mathcal S_1$ and $\mathcal S_2$ consist of two contests: the $\wta$ contest and a contest $C$ with $\cutil_C = (1, 0, 0, 0, 0, 0, 1/49, 1/64, 1/81, 1/100)$, that is, choosing Tullock contest with $\tau_k=1$ when $7\le k \le 10$ and $\tau_k=+\infty$ otherwise. Then $(C, C)$ is a contestant-symmetric subgame-perfect equilibrium, and $\wta$ is not a dominant contest for either designer.  See Example~\ref{ex:non-monotone-tullock} for details. Moreover, Example~\ref{ex:wta-not-dominant} shows that even if $S_1$ contains only contests with monotonically decreasing utility, $\wta$ may not be a dominant contest for designer 1. 
}

When the sets $\mathcal{S}_i$ contain only contests with monotonically decreasing utility, the contestant-symmetric subgame-perfect equilibria that 
Theorem~\ref{thm:max-rent-dissipation-NE} describes are the only possible equilibria:

\begin{theorem}\label{thm:uniqueness-NE}
Fix any $CCG(m, n, (R_i)_{i=1}^m, (\mathcal{S}_i)_{i=1}^m)$ where each $\mathcal{S}_i \subseteq \mathcal C_{R_i}$ only contains contests with monotonically decreasing utility. Assume $\mrd(\mathcal S_i) \ne \emptyset$ for each $i$. Pick $T_i\in\mrd(\mathcal S_i)$, and let $\tilde p_i=p_i(T_1,\ldots,T_m)$ be the probability a contestant participates in contest $T_i$ in the equilibrium of contestants, and let $P=\Supp(\bm T) = \{i:\tilde p_i>0\}$ be the set of indices of contests in which contestants participate with positive probability when the contests are $(T_1,\ldots,T_m)$. Then
\begin{enumerate}

    \item for any contestant-symmetric subgame-perfect equilibrium $(C_1,\ldots,C_m)$, $p_i(C_1,\ldots,C_m)=\tilde p_i$.

    \item if $|P|\ge2$, then $(C_1,\ldots,C_m) \in \mathcal S_1\times \cdots \times \mathcal S_m$ is a contestant-symmetric subgame-perfect equilibrium if and only if $C_i\in\mrd(\mathcal S_i), \forall i \in P$.\footnote{If $p_i(C_1,...,C_m) = 0$ then contest $i$ could be anything:
    If $p_i(C_1,\ldots,C_m) = 0$, then the utility for agent $i$ is 0, which cannot be improved by choosing any other contest $C_i'$ as $(C_i, \bm{C}_{-i})$ is an equilibrium.  Moreover, by Claim~\ref{claim:probability}, $p_i(C'_i, \bm{C}_{-i})$ must be $0$ as well, so the choice of $C'_i$ does not affect the choices of contests of other designers. 
    }
    
    \item
    if $|P|=1$, let $P=\{i_0\}$, then $(C_1,\ldots,C_m) \in \mathcal S_1\times \cdots \times \mathcal S_m$ is a contestant-symmetric subgame-perfect equilibrium if and only if $\cutil_{C_{i_0}}(n)=\cutil_{T_{i_0}}(n)$.\footnote{As $p_{i_0}(C_1,\ldots,C_m)=1$, with probability 1 there are $n$ contestants in contest $i_0$, thus the contest success functions of contest $i_0$ for $k \neq n$ have no effect on the utility calculation for the contestants' best response and could be anything.
    %
    %
    %
    }
\end{enumerate}
\end{theorem}

In the symmetric-reward case we can show that $|P|=m$ which makes the statement shorter:

\begin{corollary}
\label{cor:uniqueness-NE-symmetric-reward}
In the symmetric-reward case, i.e., $R_1 = \cdots = R_m$, $(C_1,\ldots,C_m) \in \mathcal S_1\times \cdots \times \mathcal S_m$ is a contestant-symmetric subgame-perfect equilibrium if and only if $C_i\in\mrd(\mathcal S_i)$ for all $i\in\{1, \ldots, m\}$. 
\end{corollary}
\begin{proof}
We need to show that $|P|=m$ in this symmetric-reward case. Assume by contradiction that there exists $i$ such that $\tilde p_i=0$. Since $\sum\limits_{\ell=1}^m\tilde p_\ell=1$, there exists $j \ne i$ such that $\tilde p_j>0$. Then by Claim~\ref{claim:monotone-decreasing}, $\beta(T_j,\tilde p_j)<\beta(T_j,0)=R_j=R_i=\beta(T_i,\tilde p_i)$. However, this contradicts the equilibrium condition (Claim~\ref{claim:equilibrium-condition}) which states that $\tilde p_j>0$ implies $\beta(T_j,\tilde p_j) \ge \beta(T_i,\tilde p_i)$. Therefore, we conclude that $\tilde p_i>0$ for all $i\in\{1,\ldots,m\}$, i.e., $|P|=m$ as required.
\end{proof}

Thus, the case of symmetric rewards is a ``clear cut'' while the general case is more involved. The following example demonstrates the need for this distinction using a setting with highly asymmetric rewards (see Appendix~\ref{sec:proof-example-asymmetric-reward} for a proof). 

\begin{restatable}{example}{equilibriaAsymmetricRewards}
\label{ex:equilibria-asymmetric-rewards}
Consider $m \ge 3$ contests and $n$ contestants. Contest 1 has reward $R_1=1$, and each of others has reward $R_j = \left(\frac{m-1}{m-2}\right)^{n-1}+1$. Each set $\mathcal{S}_i$ contains all monotonically decreasing utility contests (hence contains $\wta$). 
Then for any contest $C_1\in \mathcal S_1$, $(C_1, T_2, \ldots, T_m)$ where $T_j=\wta\in \mrd(\mathcal S_j)$ for $j=2, \ldots, m$ is a contestant-symmetric subgame-perfect equilibrium. 
In this equilibrium, $p_1(C_1, T_2, \ldots, T_m)=0$, and $p_j(C_1, T_2, \ldots, T_m)=\frac{1}{m-1}>0$ for any $j = 2,\ldots,m$. 
\end{restatable}



Finally, the equilibria in Theorem~\ref{thm:max-rent-dissipation-NE} are Pareto optimal for the contest designers:

\begin{definition}
~
\begin{itemize}
    \item For two strategy profiles $\hat{\vec{C}}=(\hat{C}_1,\ldots,\hat{C}_m), \vec{C}=(C_1,\ldots,C_m) \in \mathcal{S}_1\times\cdots\times\mathcal{S}_m$ of the contest competition game $CCG(m,n,(R_i)_{i=1}^m,(\mathcal S_i)_{i=1}^m)$, we say $\vec{C}$ is a \emph{Pareto improvement} of $\hat{\vec{C}}$, if 
$u_i(\bm{C}) \ge u_i(\hat{\bm{C}})$ for all $i\in\{1, \ldots, m\}$ and $u_i(\bm{C}) > u_i(\hat{\bm{C}})$ for at least one $i\in \{1, \ldots, m\}$. 

\item We say a strategy profile $\hat{\vec{C}}=(\hat{C}_1,\ldots,\hat{C}_m)$ is \emph{Pareto optimal}, if there is no Pareto improvement of it.

\end{itemize}
\end{definition}

\begin{theorem}
\label{thm:pareto-optimality}
The equilibria in Theorem~\ref{thm:max-rent-dissipation-NE} are Pareto optimal.
\end{theorem}

\section{Welfare Optimality}
\label{sec:welfare}

Throughout this section, let $\bm C = (C_1,\ldots,C_m)$ be a tuple of contests. 
Denote the sum of designers' expected utilities, the sum of contestants' expected utilities, and their sum by
\begin{equation*}
    W_D(\bm C)=\sum_{i=1}^mu_i(\bm C),\quad W_C(\bm C)=n\sum_{i=1}^mp_i\beta(C_i,p_i), \quad W_S(\bm C)=W_D(\bm C)+W_C(\bm C)
\end{equation*}
where $p_i=p_i(\bm C)$. 
By equilibrium condition (Claim \ref{claim:equilibrium-condition}) we have
$\beta(C_i, p_i)=\beta(C_j, p_j)$ for all $i,j \in \Supp(\bm C)$. Denote this constant by $u_c(\bm C)$. This is a contestant's expected utility in any contest in which she participates with positive probability. Note that by definition $p_i=0$ for any $i \notin \Supp(\bm C)$, so $\sum_{i \in \Supp(\bm C)}p_i=\sum_{i=1}^mp_i=1$. As a result,
\begin{align*}
    W_C(\bm C)
    =n\sum_{i=1}^mp_i\beta(C_i,p_i)
    =n\sum_{i \in \Supp(\bm C)}p_i\beta(C_i,p_i)
    =n\sum_{i \in \Supp(\bm C)}p_iu_c(\bm C)
    =nu_c(\bm C).
\end{align*}

For $W_S(\bm C)$, note that whenever at least one contestant participates in contest $i$, the sum of expected utilities of designer $i$ and the participants in that contest equals $R_i$. Define the random variables $k_1, ... ,k_m$ as the number of contestants in contests $C_1, \ldots, C_m$. We can represent $W_S(\bm C)$ as
\begin{align}
\label{eq:WSC}
    W_S(\bm C)
    &= \Ex[k_1, ... ,k_m]{\sum_{i=1}^mR_i\mathbbm{1}[k_i \geq 1]}
    = \sum_{i=1}^mR_i\Ex[k_i \sim \bin(n, p_i)]{\mathbbm{1}[k_i \ge 1]}\\
    &= \sum_{i=1}^mR_i\left[1-(1-p_i(\bm C))^n\right]
    = \sum_{i=1}^mR_i-\sum_{i=1}^mR_i(1-p_i(\bm C))^n.
\end{align}

Our main conclusion in this section is that the equilibria $\bm T$ of Theorem~\ref{thm:max-rent-dissipation-NE} obtain welfare optimality in several quite natural cases:

\begin{theorem}
\label{thm:total-welfare-four-cases}
Consider a contest competition game $CCG(m,n,(R_i)_{i=1}^m,(\mathcal S_i)_{i=1}^m)$ and fix some $T_i\in\mrd(\mathcal S_i)$ that also have monotonically decreasing utility. Then the equilibrium $(T_1,\ldots,T_m)$ maximizes $W_S$ in each one of the following cases:

\begin{enumerate}
    \item Unrestricted contest design: for all $i$, $\mathcal S_i = \mathcal C_{R_i}$.
    \item $\WTA$ is a possible contest: $\WTA \in \mathcal S_i$ for every $i$. (We note that $\WTA$ can be replaced with any other full rent dissipation contest.)
    \item A symmetric CCG: $R_1=\cdots=R_m=R$ and $\mathcal S_1=\cdots=\mathcal S_m=\mathcal S\subset\mathcal C_R$.
    \item An MRD-symmetric CCG: $R_1=\cdots=R_m=R$ and $\mrd(\mathcal S_1)=\cdots=\mrd(\mathcal S_m)$.
\end{enumerate}
\end{theorem}

Note that the second case generalizes the first case and the fourth case generalizes the third case (since every symmetric CCG is also MRD-symmetric). The proof of this theorem is given in Appendix~\ref{app:welfare_proofs}. The following example shows that the conclusion of Theorem~\ref{thm:total-welfare-four-cases} does not hold in general.


\begin{example}
Take $m=2,n=2,R_1=R_2=1$. Let $C, T$ be Tullock contests with $\tau_C = 1, \tau_T = 1.2$. It can be verified that  $\cutil_C=(1,\frac{1}{4})$ and $\cutil_T=(1,\frac{1}{5})$. Suppose $\mathcal S_1=\{C\}$ and $\mathcal S_2=\{C,T\}$.  Then $\mrd(\mathcal S_1) = \{C\}$ and $\mrd(\mathcal S_2) = \{T\}$. Due to symmetry, $\bm p(C,C)=(\frac{1}{2},\frac{1}{2})$, so
\[
W_S(C,C)=2-2\left(1-\frac{1}{2}\right)^2=\frac{3}{2}.
\]
We can determine $\bm p(C,T)=(\tilde p_1,1-\tilde p_1)$ by solving $\beta(C,\tilde p_1)=1-\tilde p_1+\frac{1}{4}\tilde p_1=\beta(T,1-\tilde p_1)=\tilde p_1+\frac{1}{5}(1-\tilde p_1)$ hence $\tilde p_1=\frac{16}{31}$ and $\tilde p_2=1-\tilde p_1=\frac{15}{31}$. So
\[
W_S(C,T)=2-(1-\tilde p_1)^2-(1-\tilde p_2)^2=\frac{1441}{961}<\frac{3}{2}.
\]
\end{example}




\begin{remark}
One can construct a similar example for every $m,n$, and $R_1,\ldots,R_m$. Specifically, construct $\mathcal S_i$ as follows: consider equation set
\[
R_i\sum_{k=1}^n\binom{n-1}{k-1}\tilde p_i^{k-1}(1-\tilde p_i)^{n-k}\cutil_{C_i}(k)=R,\quad \forall i=1,\ldots,m.
\]
Find a solution $(R, \{\cutil_{C_i}(k)\}_{i=1,\ldots,m,k=2,\ldots,n})$ of it, and choose a monotonically decreasing utility contest $T_i$ that has higher rent dissipation than $C_i$ for every $i$. If $T_1,\ldots,T_m$ satisfy $p_i(T_1,\ldots,T_m)\ne\tilde p_i$ for some $i$, then $W_S(T_1,\ldots,T_m)<W_S(C_1,\ldots,C_m)$, which means $(T_1,\ldots,T_m)$ does not maximize $W_S$ any more, and $\mathcal S_i=\{C_i,T_i\}$ is then a counterexample.


\end{remark}

Although the total welfare is not always maximized at the equilibria of Theorem~\ref{thm:max-rent-dissipation-NE}, it turns out that the contestants' welfare is {\em always} minimized at these equilibria outcomes, as the next theorem shows.\footnote{We already know that each such equilibrium outcome is Pareto-optimal for the designers. However, this does not immediately imply that these equilibrium outcomes minimize the contestants welfare since (1) the game is not constant-sum as $W_S(\bm C)$ depends on the $p_i$'s and (2) Pareto-optimal outcomes need not necessarily maximize the aggregate designers' utility.}

\begin{theorem}
\label{thm:contestant-welfare-mrd}
Fix some $CCG(m,n,(R_i)_{i=1}^m,(\mathcal S_i)_{i=1}^m)$.
Let $\bm T = (T_1,\ldots,T_m)$ be one of the equilibria in Theorem~\ref{thm:max-rent-dissipation-NE}, i.e., $T_i\in\mrd(\mathcal S_i)$ and $T_i$ has monotonically decreasing utility. Then for any $\bm C \in\mathcal S_1\times\cdots\times\mathcal S_m$, $W_C(\bm T) \le W_C(\bm C)$. 
\end{theorem}
\begin{proof}
Let $\tilde p_i=p_i(\bm T)$ and $p_i=p_i(\bm C)$.
We assume without loss of generality $p_1>0$ (i.e. $1 \in \Supp(\bm C)$). Then $u_c(\bm C)=\beta(C_1,p_1)$.

If $p_1 \le \tilde p_1$, then $\tilde p_1>0$, which implies $u_c(\bm T)=\beta(T_1, \tilde p_1)$. Since $T_1$ has monotonically decreasing utility, we have $\beta(T_1,\tilde p_1) \le \beta(T_1,p_1)$ by Claim~\ref{claim:monotone-decreasing}. Since $T_1$ is also the maximal rent dissipation contest of $\mathcal S_1$, we get $\beta(T_1,p_1) \le \beta(C_1,p_1)$ by Claim~\ref{claim:mrd-inequality}. These inequalities together yield
\[
u_c(\bm T)=\beta(T_1, \tilde p_1)
\le\beta(T_1,p_1)
\le\beta(C_1,p_1)=u_c(\bm C).
\]

Otherwise $p_1>\tilde p_1$, then there exists $i\in\{2,\ldots,m\}$ such that $\tilde p_i>p_i$. Note that this implies $\tilde p_i>0$, so $u_c(\bm T)=\beta(T_i,\tilde p_i)$, and, 
\[
u_c(\bm T)=
\beta(T_i,\tilde p_i)
\stackrel{T_i \text{ has MDU, Claim~\ref{claim:monotone-decreasing}}}{\le} \beta(T_i,p_i)
\stackrel{T_i\in\mrd(\mathcal S_i), \text{Claim~\ref{claim:mrd-inequality}}}{\le} \beta(C_i,p_i)
\stackrel{p_1 > 0, \text{Claim~\ref{claim:equilibrium-condition}}}{\le} \beta(C_1,p_1)
= u_c(\bm C).
\]

In either case we would get $u_c(\bm T) \le u_c(\bm C)$, hence
$W_C(\bm T)=nu_c(\bm T) \le nu_c(\bm C)=W_C(\bm C)$.
\end{proof}

Theorems~\ref{thm:total-welfare-four-cases} and~\ref{thm:contestant-welfare-mrd} together immediately imply:

\begin{corollary}
Consider a contest competition game $CCG(m,n,(R_i)_{i=1}^m,(\mathcal S_i)_{i=1}^m)$ and fix some $T_i\in\mrd(\mathcal S_i)$ that have monotonically decreasing utility. Then the equilibrium $(T_1,\ldots,T_m)$ maximizes $W_D$ in the four cases of Theorem~\ref{thm:total-welfare-four-cases}.
\end{corollary}

Thus, for example, $(\WTA, \ldots, \WTA)$ maximizes the designers welfare and minimizes the contestants welfare in the case of unrestricted contest design.

\section{Summary and Discussion}

This paper studies a complete-information competition game among contest designers. First, each designer chooses a contest. Second, each contestant chooses (possibly in a random way) in which contest to participate. Third, a symmetric equilibrium outcome is realized in each contest (contestants choose effort levels; the reward is allocated to them based on their realized efforts). The resulting utility of each contest designer is the sum of efforts invested in their respective contests. The resulting utility of each  contestant is the reward she receives minus the effort she invests.

Our main results characterize a certain type of contests which form an equilibrium (and may even be dominant) in this game. These equilibria are Pareto-optimal for the contest designers. Under natural conditions, these are the only possible equilibria. In addition, these equilibria maximize the social welfare (while minimizing the contestants' aggregate welfare) when designers are unrestricted in their choice of a contest, or when they are all restricted in the same way. Our results yield several conclusions regarding Tullock contests. For example, if contest designers are restricted to choose a Tullock contest with some parameter $\tau$ then any Tullock contest  with $\tau \geq 2$ (e.g., $\wta$ where $\tau=\infty$) is a dominant contest for every designer.

It may be interesting for future work to further examine and relax several of our assumptions:

\begin{trivlist}

\item {\bf The Participation Model:} We make several assumptions that (although natural in many cases) could be further relaxed. First, one may assume that contestants cannot observe the total number of contestants in the contest they chose which may be the reality in large electronic/online contests (but is less realistic in small physical contests). Second, a contestant may be able to participate in more than one contest simultaneously, e.g., in up to some fixed maximal number of contests.
Another option is to study a budgeted participation model where the total effort of each contestant could be split among several contests (see e.g.~\cite{LS20}). Third, if we allow contestants to choose an asymmetric equilibrium in response to the designers' contest success functions, our results no longer hold. Example~\ref{ex:asymmetric} in the Appendix shows that $(\wta,\wta)$ is not an equilibrium when we have two Tullock contests and three contestants that may choose an asymmetric equilibrium for their participation probabilities. It can be interesting to understand the asymmetric case as well.

\item {\bf Stochastic Ability:}
When a contestant has a cost of effort that depends on a stochastic (and private) ability,
the single contest game between a designer and $k\ge 1$ contestants is not constant-sum. Importantly, our argument in this paper uses its constant-sum property when contestants' abilities are a fixed constant; this property allows us to relate the designer's utility with contestants' utilities which can then be characterized by the equilibrium condition for the contestants.  When the game is not constant-sum, our argument cannot be applied directly.  Another obstacle to analyzing stochastic ability is the lack of explicit characterization of contestants' utilities in a contest (even in a single Tullock contest) with stochastic ability.

\item {\bf Non-MDU contests: } When the prize structure can depend on the number of participants, e.g., choosing a different Tullock contest depending on the number of participants, the MDU property will typically be violated, see e.g., Example~\ref{ex:non-MDU}. This opens the possibility of new ``exotic'' contests with new types of equilibria that do not fall within Theorem~\ref{thm:max-rent-dissipation-NE}. 

\item {\bf Risk-averse contestants: } In Appendix~\ref{app:risk-averse} we show an example with two contest designers that can choose some Tullock contest with $\tau \in [0,2]$. With risk-neutral contestants we showed that $(\tau=2,\tau=2)$ is the unique contestant-symmetric subgame-perfect equilibrium (Corollary~\ref{cor:full-rent-dissipation-NE} and Corollary~\ref{cor:uniqueness-NE-symmetric-reward}).
The example shows that with risk-averse contestants $(\tau=2,\tau=2)$ is no longer a contestant-symmetric subgame-perfect equilibrium. It can be interesting to investigate how contestants' risk perceptions affect the qualitative conclusions of this work. In particular, we conjecture that with risk-loving participants, our results continue to hold.

\end{trivlist}

\appendix

\section{Missing Proofs from Section~\ref{sec:model}}
\label{app:model_proofs}

\subsection{Proof of Claim~\ref{claim:monotone-decreasing}}
\monotonedecreasing*

\begin{proof}
By definition, 
\begin{align*}
    \beta(C_i, p) & = (1-p)^{n-1} R_i + \sum_{k=1}^{n-1} \binom{n-1}{k} p^k (1-p)^{n-1-k} \cutil_{C_i}(k+1) \\
    & = (1-p)^{n-1} \frac{R_i}{2}  +  (1-p)^{n-1} \frac{R_i}{2} + \sum_{k=1}^{n-1} \binom{n-1}{k} p^k (1-p)^{n-1-k} \cutil_{C_i}(k+1). 
\end{align*}
Let
\[ \tilde{\cutil}(k+1) = \begin{cases} \frac{R_i}{2} & k=0; \\ \cutil_{C_i}(k+1) & \text{otherwise.} \end{cases} \]
We can write $\beta(C_i, p) = (1-p)^{n-1} \frac{R_i}{2}  + \Ex[k\sim \bin(n-1, p)]{ \tilde{\cutil}(k+1)}$. 
We recall that $\cutil_{C_i}(k)$, the utility of a contestant in a contest with $k$ contestants, is at most $\frac{R_i}{k}$; in particular, $\cutil_{C_i}(2) \le \frac{R_i}{2}$. Therefore, the sequence
\[ \tilde{\cutil}(1) = \frac{R_i}{2} \ge \tilde{\cutil}(2) \ge \cdots \ge \tilde{\cutil}(n) \]
is decreasing under the assumption that $C_i$ has monotonically decreasing utility. Due to the first order stochastic dominance of a binomial distribution with a higher $p$ parameter over another binomial distribution with a lower $p$ parameter (see e.g., \citealp{Wolfstetter:99}), we have
\[ \Ex[k\sim \bin(n-1, p)]{ \tilde{\cutil}(k+1)} \ge \Ex[k\sim \bin(n-1, p')]{ \tilde{\cutil}(k+1)}  \]
which implies
\begin{align*}
    \beta(C_i, p) & = (1-p)^{n-1}\frac{R_i}{2} + \Ex[k\sim \bin(n-1, p')]{ \tilde{\cutil}(k+1)} \\
     & > (1-p')^{n-1}\frac{R_i}{2} + \Ex[k\sim \bin(n-1, p')]{ \tilde{\cutil}(k+1)} = \beta(C_i, p'). 
     \qedhere
\end{align*}
\end{proof}


\subsection{Proof of Lemma~\ref{lem:tullock}}
\tullock*

\begin{proof}
Proposition 2 of \cite{segev_schweinzer_2012} shows that for any $k\ge 2$ such that $\frac{k}{k-1} \geq \tau$, $\gamma_{C_{\tau}}(k) = R( \frac{1}{k} - \frac{k-1}{k^2} \tau)$; in particular, $0\le \gamma_{C_{\tau}}(k) \le \frac{1}{k}R < R = \gamma_{C_\tau}(1)$. 
Corollary 5 of \cite{ewerhart_2017_mdu} shows that if $\frac{k}{k-1} < \tau$, then $\gamma_{C_{\tau}}(k) = 0$. \cite{baye1996apa} show that for $\tau=+\infty$, $\cutil_{C_{\tau}}(k)=0$ for $k\ge 2$. 

To prove the third corollary of the lemma, we note that for any fixed $k$, $\cutil_{C_\tau}(k) = R( \frac{1}{k} - \frac{k-1}{k^2} \tau)$ is a non-negative decreasing function of $\tau$ for $0\le \tau \le \frac{k}{k-1}$, and $\cutil_{C_\tau}(k) = 0$ for $\tau > \frac{k}{k-1}$, 
so $\cutil_{C_{\tau^{\max}}}(k) \le \cutil_{C_{\tau'}}(k)$ for any $\tau' \le \tau^{\max}$. 

We now turn our attention to the first two corollaries. 
If $\tau > 2$, then $\frac{k}{k-1} < \tau$ for any $k \ge 2$, so $C_\tau$ has full rent dissipation. 
If $\tau = 2$, then for $k=2$ we $\gamma_{C_{\tau}}(k) = R(\frac 1 k - \frac{k-1}{k^2}\tau) = 0$, and for any $k\geq 3$ we have $\gamma_{C_{\tau}}(k) = 0$ since $\frac{k}{k-1} < \tau$, so $C_\tau$ has full rent dissipation.
For any $\tau < 2$, we let $K_{\tau}$ be the first $k\in \{2,3,...\}$ such that $\frac{K_{\tau}}{K_{\tau}-1} < \tau$ (let $K_{\tau} = \infty$ for $\tau \le 1$). In order to show that $C_\tau$ has monotonically decreasing utility, we only need to show that $\gamma_{C_{\tau}}(k)$ is monotonically decreasing in $k$ in the range $2\leq k < K_{\tau}$.  Consider the function $f(x) = \frac{1}{x} - \frac{x-1}{x^2} \tau$ for $2\le x < K_{\tau}$.  We take its derivative:
\[
 f'(x) = -\frac{1}{x^2} - \frac{x^2 - 2(x-1)x}{x^4} \tau = - \frac{1}{x^2} + \frac{x-2}{x^3}\tau. 
\]
Because $\tau \le \frac{x}{x-1}$ for $x< K_\tau$, we have 
\[
 f'(x) \le - \frac{1}{x^2} + \frac{x-2}{x^3} \frac{x}{x-1} = \frac{1}{x^2}\left(-1 + \frac{x-2}{x-1}\right) = -\frac{1}{x^2} \cdot \frac{1}{x-1} \le 0. 
\]
Thus, $f(x)$ is monotonically decreasing, and so is $\gamma_{C_{\tau}}(k)$. 
\end{proof}

\section{Missing Proofs from Section~\ref{sec:main_results}}
\label{app:main_results_proofs}

Throughout this section we assume that $\mathcal S_i \subseteq \mathcal C_{R_i}$ contains a maximal rent dissipation contest with monotonically decreasing utility, which we denote by $T_i \in \mrd(\mathcal S_i)$.

\subsection{Analysis of the contest competition game: proof of Theorem~\ref{thm:max-rent-dissipation-NE}} \label{sec:designer-game-proof-of-lemma}

Theorem~\ref{thm:max-rent-dissipation-NE} immediately follows from the following lemma, which is our main technical lemma.  It shows that for any designer $i$, choosing $T_i$ is always a best response if other designers choose contests with monotonically decreasing utility. 
\begin{lemma}
\label{lem:general_dominance}
Fix any $CCG(m, n, (R_i)_{i=1}^m, (\mathcal S_i)_{i=1}^m)$ where for any $i$, $\mathcal S_i \subseteq \mathcal C_{R_i}$ contains a maximal rent dissipation contest $T_i$ that has monotonically decreasing utility.
Fix some designer $i$ and for all $j \neq i$ fix some $\hat C_j \in \mathcal S_j$ with monotonically decreasing utility. Then $u_i(T_i, \hat{\bm C}_{-i}) \ge u_i(C_i, \hat{\bm C}_{-i})$ for all $C_i\in\mathcal S_i$.  
\end{lemma}




We prove a useful claim before proving the lemma.
The claim is about the change of contestants' participation equilibrium when a designer switches her contest.
Intuitively, we expect that, with all other contests left unchanged, a designer (say, designer 1) setting her contest to have less rent dissipation results in higher participation (properties (\ref{claim:probability-p1}) and (\ref{claim:probability-pj}) in the next claim). A more surprising property is that if participants are unwilling to participate in contest $1$ under maximal rent dissipation, then no other contest will be lucrative enough to attract them (property (\ref{claim:probability-pi=0})). 

\begin{claim}\label{claim:probability}
Following the notations in Lemma~\ref{lem:general_dominance}, we let $(\hat p_1, \hat p_2, \ldots, \hat p_m) =\bm{p}(T_1, \hat C_2, \ldots, \hat C_m)$, and for any $C_1 \in \mathcal S_1$, let 
$(p_1, p_2, \ldots, p_m) = \bm{p}(C_1, \hat C_2, \ldots, \hat C_m)$. 
Then,
\begin{enumerate}[(a)]
    \item $p_1 \ge \hat p_1$; \label{claim:probability-p1}
    \item $p_j \le \hat p_j$ for all $j\in\{2, \ldots, m\}$; \label{claim:probability-pj}
    \item $\hat p_i = 0 \implies p_i=0$ for all $i\in \{1, 2, \ldots, m\}$.\label{claim:probability-pi=0}
\end{enumerate}
\end{claim}

\begin{proof}
\textbf{Proof of (\ref{claim:probability-p1}).} Assume by contradiction $\hat{p}_1 > p_1 \geq 0$, then there is some designer $j\ne 1$ with $\hat p_j < p_j$, because the probabilities sum to one. We now have the contradiction: 
\begin{align*}
    & \beta(T_1, \hat p_1) \ge && (\hat p_1 > 0, \text{ Claim~\ref{claim:equilibrium-condition}}) 
     \\
     & \beta(\hat C_j, \hat p_j) > && (\hat p_j < p_j, \hat C_j \text{ has MDU, Claim~\ref{claim:monotone-decreasing}} ) \\
     & \beta(\hat{C}_j,p_j) \ge && ( p_j > 0, \text{ Claim~\ref{claim:equilibrium-condition}} ) \\
     & \beta(C_1, p_1) \ge && (T_1\in \mrd(\mathcal S_1)\text{, Claim~\ref{claim:mrd-inequality}}) \\
     & \beta(T_1,p_1) > && (p_1<\hat p_1,T_1 \text{ has MDU, Claim~\ref{claim:monotone-decreasing}}) \\
    & \beta(T_1,\hat{p}_1).
\end{align*}

\noindent
\textbf{Proof of (\ref{claim:probability-pj}).} From (\ref{claim:probability-p1}) we know
\begin{equation}
\label{eq:sum_of_p_increase}
\sum\limits_{j=2}^m \hat{p}_j =  1 - \hat p_1\ge 1 - p_1 = \sum\limits_{j=2}^mp_j. 
\end{equation}
Assume by contradiction that for some $j_0\in\{2,3,\ldots,m\}$, $p_{j_0}>\hat p_{j_0}\geq0$, then by \eqref{eq:sum_of_p_increase} there must exist $j\in\{2,3,\ldots,m\}$ with $j\ne j_0$ such that $\hat p_j>p_j\ge0$. We therefore have the contradiction
\begin{align*}
    & \beta(\hat C_j,p_j) > && (p_j<\hat p_j,\hat C_j\text{ has MDU, Claim~\ref{claim:monotone-decreasing}}) \\
    & \beta(\hat C_j,\hat p_j) \ge && (\hat p_j>0, \text{Claim~\ref{claim:equilibrium-condition}}) \\
    & \beta(\hat C_{j_0},\hat p_{j_0}) > && (\hat p_{j_0}<p_{j_0},\hat C_{j_0}\text{ has MDU, Claim~\ref{claim:monotone-decreasing}}) \\
    & \beta(\hat C_{j_0},p_{j_0}) \ge && (p_{j_0}>0, \text{Claim~\ref{claim:equilibrium-condition}}) \\
    & \beta(\hat C_j,p_j).
\end{align*}

\noindent
\textbf{Proof of (\ref{claim:probability-pi=0}).} For any $i\in \{2, \ldots, m\}$, (c) is an immediate corollary of (b). Consider $i=1$, and assume by contradiction that $p_1 > \hat{p}_1 = 0$. Then there is some designer $j\ne 1$ with $p_j < \hat{p}_j$
because the probabilities sum to one. We then have the contradiction

\begin{align*}
    & \beta(T_1, \hat p_1) = && (\hat p_1 = 0) \\
    & \cutil_{T_1}(1) = && ( \text{By definition}) \\
    & R_1 \geq && (\cutil_{C_1}(k+1) \le \frac{R_1}{k+1} \le R_1 \text{ by Eq.~\eqref{eq:cutil-R-k}}) \\
    & \Ex[k\sim \bin(n-1, p_1)]{\cutil_{C_1}(k+1)} = \beta(C_1, p_1) \ge && (p_1>0, \text{Claim~\ref{claim:equilibrium-condition}}) \\
    & \beta(\hat{C}_j,p_j) > && (p_{j}< \hat p_{j},\hat C_{j}\text{ has MDU, Claim~\ref{claim:monotone-decreasing}}) \\
    & \beta(\hat{C}_j,\hat{p}_j) \geq && (\hat{p}_j>0, \text{Claim~\ref{claim:equilibrium-condition}}) \\
    & \beta(T_1,\hat{p}_1). && \hfill \qedhere
\end{align*}

\end{proof}

\noindent
\textbf{Proof of Lemma \ref{lem:general_dominance}.}  Without loss of generality, we only prove it for designer $i=1$.  
Suppose that when designer $1$ chooses $T_1$ and all other designers $j$ choose some contests $\hat C_j$ with monotonically decreasing utility, contestants choose participation probabilities $(\hat p_1, \hat p_2, \ldots, \hat p_m) = \bm{p}(T_1, \hat C_2, \ldots, \hat C_m)$.  According to Claim~\ref{claim:designer-utility}, the expected utility of designer 1, denoted by $\hat u_1$, equals 
\begin{equation}\label{eq:designer-hat-utility}
    \hat u_1 = u_1(T_1, \hat C_2, \ldots, \hat C_m) = R_1 \left[ 1 - (1- \hat p_1)^n\right] - n\hat p_1 \beta(T_1, \hat p_1). 
\end{equation}
When designer 1 switches to any other contest $C_1\in \mathcal S_1$, letting $(p_1, p_2, \ldots, p_m)=\bm{p}(C_1, \hat C_2, \ldots, \hat C_m)$, the expected utility of designer 1 becomes
\begin{equation}\label{eq:designer-non-hat-utility}
    u_1 = u_1(C_1, \hat C_2, \ldots, \hat C_m) = R_1 \left[ 1 - (1- p_1)^n\right] - n p_1 \beta( C_1,  p_1). 
\end{equation}
Our goal is to show that $\hat u_1 \ge u_1$. 

If $\hat p_1 = 0$, then by (\ref{claim:probability-pi=0}) of Claim~\ref{claim:probability} we have $p_1 = 0$ and hence $\hat u_1 = \hat u_1 = 0$.  The conclusion holds. 

Now assume $\hat p_1 > 0$. If $\hat p_j = 0$ for all $j\in\{2, \ldots, m\}$, then by (\ref{claim:probability-pi=0}) of Claim~\ref{claim:probability} we have $p_j = 0$ for all $j$.  This implies $\hat p_1 = p_1 = 1$ and hence $\hat u_1 = R_1 - n \cutil_{T_1}(n)$ and $u_1 = R_1 -  n \cutil_{C_1}(n)$.  Since $\cutil_{T_1}(n) \le \cutil_{C_1}(n)$ by the assumption that $T_1$ has maximal rent dissipation, we have $\hat u_1 \ge u_1$. 

Now we consider the case where $\hat p_j>0$ for some $j\in \{2, \ldots, m\}$.  Because each contestant participates in both $T_1$ and $\hat C_j$ with positive probability, by equilibrium condition (Claim~\ref{claim:equilibrium-condition}), we must have 
\begin{equation*}
\beta(T_1, \hat p_1) = \beta(\hat C_j, \hat p_j).
\end{equation*}
By (\ref{claim:probability-p1}) of Claim~\ref{claim:probability}, $p_1 \ge \hat p_1 >0$, so each contestant participates in $C_1$ with positive probability, and by equilibrium condition (Claim~\ref{claim:equilibrium-condition}),
\begin{equation*}
    \beta(C_1, p_1) \ge \beta(\hat C_j, p_j).
\end{equation*}
According to Claim~\ref{claim:monotone-decreasing},  $\beta(\hat C_j, p)$ is a monotonically decreasing function of $p$, and by (\ref{claim:probability-pj}) of Claim~\ref{claim:probability}, $p_j \le \hat p_j$.  Therefore, we have $\beta( \hat C_j, p_j) \ge \beta( \hat C_j, \hat p_j)$ and hence \begin{equation}\label{eq:contestant-utility-chain}
    \beta(C_1, p_1) \ge \beta(\hat C_j, p_j) \ge \beta(\hat C_j, \hat p_j) = \beta(T_1, \hat p_1). 
\end{equation}
Plugging \eqref{eq:contestant-utility-chain} into \eqref{eq:designer-non-hat-utility}, we get
\[ u_1 \le R_1[1-(1- p_1)^n] - n p_1 \beta( T_1, \hat p_1). \]

Now we define function
\begin{equation}
\label{eq:f(p)}
    f(p) = R_1[1-(1- p)^n] - n p \beta(T_1, \hat p_1).
\end{equation}
We take its derivative: 
\begin{align*}
    f'(p) & = n R_1 (1-p)^{n-1} - n \beta(T_1, \hat p_1) \\
& = nR_1(1-p)^{n-1}  - n \sum\limits_{k=0}^{n-1} \binom{n-1}{k} \hat p_1^k (1- \hat p_1)^{n-1-k} \cutil_{T_1}(k+1) \\
& = nR_1(1-p)^{n-1}  - n(1 - \hat p_1)^{n-1} R_1 -  n \sum\limits_{k=1}^{n-1} \binom{n-1}{k} \hat p_1^k (1- \hat p_1)^{n-1-k} \cutil_{T_1}(k+1) \\
& \le nR_1(1-p)^{n-1}  - nR_1 (1 - \hat p_1)^{n-1}. 
\end{align*}
For $p > \hat p_1$, $(1-p)^{n-1} < (1-\hat p_1)^{n-1}$, so $f'(p) < 0$.  Thus, $f(p)$ is monotonically decreasing in the range $[\hat p_1, 1]$, which implies
\begin{equation}\label{eq:u1-comparison}
    \hat u_1 = f(\hat p_1) \ge f(p_1) \ge u_1,
\end{equation}
concluding the proof. 

\subsection{Full characterization of equilibria for MDU contests: proof of Theorem~\ref{thm:uniqueness-NE}}

\noindent
{\bf Additional properties of contestants' participation game.}
The following two claims use the notation and assumptions of Lemma~\ref{lem:general_dominance}, specifically, we fix a contest competition game $CCG(m,n,(R_i)_{i=1}^m,(\mathcal S_i)_{i=1}^m)$, where for every $i$, $\mrd(\mathcal S_i)$ contains at least one contest $T_i$, and $T_i$ has monotonically decreasing utility.
The first claim and its proof are similar to item (\ref{claim:probability-pi=0}) of Claim \ref{claim:probability}.\footnote{We note the differences: (1) here every contest changes, while in Claim~\ref{claim:probability} only one contest changes, and (2) here we require all contests to be maximal rent dissipation (and MDU), while in Claim~\ref{claim:probability} we only require them to be MDU.}

\begin{claim}
\label{claim:participation-probabilities-po}
Let $\tilde p_i=p_i(T_1, \ldots, T_m)$. For any strategy profile $(C_1, \ldots, C_m)\in\mathcal S_1 \times \cdots \times \mathcal S_m$, let $p_i=p_i(C_1, \ldots, C_m)$. Then for any $i\in\{1, \ldots, m\}$, $\tilde p_i=0$ implies $p_i=0$.
\end{claim}
\begin{proof}
 Assume by contradiction there exists $i$ such that $\tilde p_i=0$ and $p_i>0$. Then there exists $j \ne i$ such that $\tilde p_j>p_j$. By Claim \ref{claim:equilibrium-condition} we have $\beta(T_j,\tilde p_j) \ge \beta(T_i,\tilde p_i)=R_i$ and $\beta(C_i,p_i) \ge \beta(C_j,p_j)$.
 By Claim~\ref{claim:mrd-inequality} (since $T_j \in \mrd(\mathcal S_j)$), we have $\beta(C_j,p_j) \ge \beta(T_j,p_j)$.
 By Claim~\ref{claim:monotone-decreasing}, since $T_j$ is a MDU contest and $\tilde p_j>p_j$, we have $\beta(T_j,p_j)>\beta(T_j,\tilde p_j)$. Finally, it is obvious that $\beta(C_i,p_i) \le R_i$. Combining these inequalities, we get
\[
\beta(T_j,\tilde p_j)
\ge\beta(T_i,\tilde p_i)=R_i
\ge\beta(C_i,p_i)
\ge\beta(C_j,p_j)
\ge\beta(T_j,p_j)
>\beta(T_j,\tilde p_j),
\]
which is a contradiction.
\end{proof}

When $\hat{\bm{C}}_{-i}$ are all MDU contests, Lemma~\ref{lem:general_dominance} states that $T_i$ is a dominant contest for designer $i$. Thus, if we have $u_i(T_i, \hat{\bm C}_{-i}) = u_i(C_i, \hat{\bm C}_{-i})$ for some $C_i\in\mathcal S_i$, then $C_i$ is a best response to $\hat{\bm C}_{-i}$. The next claim shows that, in this case, the equilibrium outcome in the two contestants' participation games $(T_i, \hat{\bm C}_{-i})$ and $(C_i, \hat{\bm C}_{-i})$ is identical.

\begin{claim}\label{claim:mutual-dominate}
If $u_i(T_i, \hat{\bm C}_{-i}) = u_i(C_i, \hat{\bm C}_{-i})$ for some $C_i\in\mathcal S_i$, then $\hat p_i=p_i$ and $\beta(T_i, \hat p_i) = \beta(C_i, p_i)$, where $\hat p_i=p_i(T_i,\hat{\vec C}_{-i})$ and $p_i=p_i(C_i, \hat{\vec C}_{-i})$.

\end{claim}
\begin{proof}
Without loss of generality, we only prove it for contest designer $i=1$. 
Following the notation in Section~\ref{sec:designer-game-proof-of-lemma} (Eq.~\eqref{eq:designer-hat-utility}, Eq.~\eqref{eq:designer-non-hat-utility} and Eq.~\eqref{eq:f(p)}), define
\begin{align*}
    & \hat u_1 = u_1(T_1, \hat{\bm C}_{-1}) = R_1 \left[ 1 - (1- \hat p_1)^n\right] - n\hat p_1 \beta(T_1, \hat p_1),\\
    & u_1 = u_1(C_1, \hat{\bm C}_{-1}) = R_1 \left[ 1 - (1- p_1)^n\right] - n p_1 \beta( C_1,  p_1),\\
    & f(p) = R_1[1-(1- p)^n] - n p \beta(T_1, \hat p_1).
\end{align*}
Recall that by the assumption in the statement of the claim we have $\hat u_1=u_1$. Consider the following three cases:
\begin{itemize}
    \item 
    If $\hat p_1=0$, then by \eqref{claim:probability-pi=0} of Claim~\ref{claim:probability}, $p_1=0=\hat p_1$. And $\beta(T_1,\hat p_1)=\beta(C_1,p_1)=R_1$.
    \item
    If $\hat p_1=1$, then by \eqref{claim:probability-p1} of Claim~\ref{claim:probability}, $1 \geq p_1 \ge \hat p_1=1$ hence $p_1=\hat p_1=1$. Therefore,
    \begin{align*}
        R_1-n\beta(T_1,\hat p_1)=\hat u_1=u_1 =R_1-n\beta(C_1,p_1).
    \end{align*}
    This immediately implies $\beta(T_1,\hat p_1)=\beta(C_1,p_1)$.
    \item
    Otherwise, $0<\hat p_1<1$, so there exists some $j\in\{2,\ldots,m\}$ such that $\hat p_j>0$ and therefore Eq.~\eqref{eq:u1-comparison} in the proof of Lemma~\ref{lem:general_dominance} holds. Furthermore, since $\hat u_1=u_1$, all the inequalities in \eqref{eq:u1-comparison} become equalities. Thus,
    \begin{equation}\label{eq:ineq-become-eq-1}
        f(\hat p_1)=f(p_1) = u_1.
    \end{equation}
    
By \eqref{claim:probability-p1} of Claim~\ref{claim:probability}, $p_1 \ge \hat p_1$, so by strict monotonicity of $f(p)$ in the range $p\in[\hat p_1,1]$, Eq.~\eqref{eq:ineq-become-eq-1} implies $p_1=\hat p_1$ and in addition
\begin{equation*}
    R_1[1-(1-p_1)^n]-np_1\beta(T_1,\hat p_1)=f(p_1) =u_1=R_1[1-(1-p_1)^n]-np_1\beta(C_1,p_1),
\end{equation*}
which directly implies
    $\beta(T_1,\hat p_1)=\beta(C_1,p_1)$ since $p_1=\hat p_1>0$. \qedhere
\end{itemize}
%
\end{proof}

\noindent
{\bf Proof of the ``$\Longrightarrow$'' direction of items 2 and 3 of Theorem~\ref{thm:uniqueness-NE}.}
Under the notation of Theorem~\ref{thm:uniqueness-NE}, for every contestant-symmetric subgame-perfect equilibrium $(C_1,...,C_m)$, recall that $\Supp(C_1,\ldots,C_m)=\{i:p_i(C_1,\ldots,C_m)>0\}$. We defer the proof that $\Supp(C_1,\ldots,C_m)=P$ and first prove a useful lemma:

\begin{lemma}\label{lem:necessity-of-uniqueness}
~~
\begin{itemize}
    \item If $|\Supp(C_1,\ldots,C_m)|>1$ then for any $i \in \Supp(C_1,\ldots,C_m)$, $C_i \in \mrd(\mathcal S_i)$.

    \item If 
    $\Supp(C_1,\ldots,C_m)=\{i_0\}$ then $p_{i_0}(C_1,\ldots,C_m)=1$ and $\cutil_{C_{i_0}}(n)=\cutil_{T_{i_0}}(n)$.
\end{itemize}
\end{lemma}

\begin{proof}
Recall that $T_i$ is a dominant contest and hence a best response to $\bm{C}_{-i}$.  Since $(C_1, \ldots, C_m)$ is a contestant-symmetric subgame-perfect equilibrium, $C_i$ is also a best response.
Applying Claim~\ref{claim:mutual-dominate}, we get $p_i:=p_i(C_i,\vec C_{-i})=p_i(T_i,\vec C_{-i})$ and
\begin{equation}\label{eq:eq-of-cutil}
    \beta(T_i,p_i)=\beta(T_i,p_i(T_i,\vec C_{-i}))=\beta(C_i,p_i(C_i,\vec C_{-i}))=\beta(C_i,p_i).
\end{equation}
By definition
\begin{align*}
    \beta(T_i,p_i)
    &= \sum\limits_{k=0}^{n-1}\binom{n-1}{k}p_i^k(1-p_i)^{n-1-k}\cutil_{T_i}(k+1)\\
    =\beta(C_i,p_i)
    &= \sum\limits_{k=0}^{n-1}\binom{n-1}{k}p_i^k(1-p_i)^{n-1-k}\cutil_{C_i}(k+1).
\end{align*}

If $p_i=1$, then $\beta(T_i,p_i)=\cutil_{T_i}(n)=\beta(C_i,p_i)=\cutil_{C_i}(n)$, this corresponds to the second case of the lemma.

Otherwise, for any $i\in\{1,\ldots,m\}$, $p_i<1$. We prove the first case of the lemma, i.e. for any $i \in \Supp(C_1,\ldots,C_m)$, $C_i \in \mrd(\mathcal S_i)$. Actually, as $0<p_i<1$, $\binom{n-1}{k}p_i^k(1-p_i)^{n-1-k}>0$ for every $k=0,\ldots,n-1$. Moreover, $T_i \in \mrd(\mathcal S_i)$ implies $\cutil_{T_i}(k+1)\le\cutil_{C_i}(k+1)$ for every $k=0,\ldots,n-1$. As a result, for \eqref{eq:eq-of-cutil} to hold, we must have $\cutil_{T_i}(k+1)=\cutil_{C_i}(k+1)$ for every $k=0,\ldots,n-1$, which indicates $C_i \in \mrd(\mathcal S_i)$. This completes the proof of the lemma.
\end{proof}

Comparing Lemma~\ref{lem:necessity-of-uniqueness} with the conclusion of the ``$\implies$'' direction, we are left to prove $\Supp(C_1,\ldots,C_m)=P$. The $\Supp(C_1,\ldots,C_m) \subseteq P$ result is just a direct implication of Claim~\ref{claim:participation-probabilities-po}.
We then prove $P \subseteq \Supp(C_1,\ldots,C_m)$.
Denote $p_i(C_1,\ldots,C_m)$ by $p_i$ for simplicity. 
Assume towards a contradiction that there exists $i\in\{1,\ldots,m\}$ such that $\tilde p_i>p_i=0$. Then there exists $j \ne i$ such that $p_j>\tilde p_j$. Note that this implies $p_j>0$. Therefore, by Lemma~\ref{lem:necessity-of-uniqueness}, either $C_j \in \mrd(\mathcal S_j)$, or $p_j=1$ and $\cutil_{C_j}(n)=\cutil_{T_j}(n)$. In either case we have $\beta(T_j,p_j)=\beta(C_j,p_j)$.
By equilibrium condition (Claim~\ref{claim:equilibrium-condition}), $\beta(C_j,p_j) \stackrel{(p_j>0)}{\ge} \beta(C_i,p_i)\stackrel{(p_i=0)}{=}R_i$ and $\beta(T_i,\tilde p_i) \stackrel{(\tilde p_i > 0)}{\ge} \beta(T_j,\tilde p_j)$. By Claim~\ref{claim:monotone-decreasing}, $\beta(T_j,\tilde p_j) > \beta(T_j,p_j)$ and $R_i=\beta(T_i,p_i)>\beta(T_i,\tilde p_i)$. These inequalities together yield
\[
\beta(T_j,\tilde p_j)
>\beta(T_j,p_j)
=\beta(C_j,p_j)
\ge\beta(C_i,p_i)
=R_i
=\beta(T_i,p_i)
>\beta(T_i,\tilde p_i)
\ge\beta(T_j,\tilde p_j),
\]
which is a contradiction. Therefore, we conclude that $p_i=0$ implies $\tilde p_i=0$ for any $i\in\{1,\ldots,m\}$. This completes the proof.

\vspace{3mm}

\noindent
{\bf Proof of the ``$\Longleftarrow$'' direction of items 2 and 3 of Theorem~\ref{thm:uniqueness-NE}.}
To prove this direction we first assume $|P|>1$. Assume $(C_1,\ldots,C_m)$ is any strategy profile satisfying $C_i \in \mrd(\mathcal S_i)$ for any $i \in P$. To prove that it is a contestant-symmetric subgame-perfect equilibrium, we only need to show that for any $i$, $C_i$ is designer $i$'s best response when the other designers choose $\bm{C}_{-i}$. Note that Lemma~\ref{lem:general_dominance} already guarantees that for $i \in P$, $C_i$ is designer $i$'s best response, so we are left to show that this also holds for those $i \notin P$. Assume $\tilde p_i=0$, then for any $C_i'\in\mathcal S_i$, by Claim~\ref{claim:participation-probabilities-po}, $p_i(C_i',\bm{C}_{-i})=0$, which implies that designer $i$ gets zero utility no matter which $C_i'$ she chooses. So $C_i$ is indeed one of her best responses. To conclude, $C_i$ is designer $i$'s best response for any $i$, which implies that $(C_1,\ldots,C_m)$ is a contestant-symmetric subgame-perfect equilibrium.

We then assume $|P|=1$. Suppose $P=\{i_0\}$, and assume $(C_1,\ldots,C_m)$ is any strategy profile satisfying $\cutil_{C_{i_0}}(n)=\cutil_{T_{i_0}}(n)$. We need to show that for any $i$, $C_i$ is designer $i$'s best response when the other designers choose $\bm{C}_{-i}$. This time Claim~\ref{claim:participation-probabilities-po} promises that for any $i \ne i_0$ and any $C_i' \in \mathcal S_i$, $p_i(C_i',\bm{C}_{-i})=0$, so $C_i$ is designer $i$'s best response. And by the same claim, $p_i(C_{i_0},\bm{C}_{-i_0})=p_i(T_{i_0},\bm{C}_{-i_0})=0$ for any $i \ne i_0$, so $p_{i_0}(C_{i_0},\bm{C}_{-i_0})=p_{i_0}(T_{i_0},\bm{C}_{-i_0})=1$. As a result,
\[
\beta(C_{i_0},p_{i_0}(C_{i_0},\bm{C}_{-i_0}))=\cutil_{C_{i_0}}(n)=\cutil_{T_{i_0}}(n)=\beta(T_{i_0},p_{i_0}(T_{i_0},\bm{C}_{-i_0})),
\]
and
\[
u_{i_0}(C_{i_0},\bm{C}_{-i_0})=R_{i_0}-n\beta(C_{i_0},p_{i_0}(C_{i_0},\bm{C}_{-i_0}))=R_{i_0}-n\beta(T_{i_0},p_{i_0}(T_{i_0},\bm{C}_{-i_0}))=u_{i_0}(T_{i_0},\bm{C}_{-i_0}).
\]
In other words, $C_{i_0}$ has equal utility for designer $i_0$ as her best response $T_{i_0}$, which implies that $C_{i_0}$ is also designer $i_0$'s best response. To conclude, $C_i$ is designer $i$'s best response for any $i$, so $(C_1, \ldots, C_m)$ is a contestant-symmetric subgame-perfect equilibrium. This completes the proof.

\vspace{3mm}

\noindent
{\bf Proof of item 1 of Theorem~\ref{thm:uniqueness-NE}.}
For any contestant-symmetric subgame-perfect equilibrium $(C_1, \ldots, C_m)$, we claim that $(\tilde p_1,\ldots,\tilde p_m)$ is a symmetric equilibrium for the contestants under $(C_1,\ldots,C_m)$; then, since the contestants' symmetric equilibrium is unique according to Lemma~\ref{lem:MDU-contestant-unique}, we must have $p_i(C_1, \ldots, C_m)=\tilde p_i$, which completes the proof. 
Consider $\beta(C_i,\tilde p_i)$ for all $i \in \{1,\ldots,m\}$. If $i \notin P$, then $\tilde p_i=0$ and $\beta(C_i,\tilde p_i)=R_i=\beta(T_i,\tilde p_i)$. If $i \in P$, then by item 2 and 3, either $C_i\in\mrd(\mathcal S_i)$ or $\tilde p_i=1$ and $\cutil_{C_i}(n)=\cutil_{T_i}(n)$, and we have $\beta(C_i,\tilde p_i)=\beta(T_i,\tilde p_i)$ in either case. So $\beta(C_i,\tilde p_i)=\beta(T_i,\tilde p_i)$ for all $i \in \{1,\ldots,m\}$. Then applying equilibrium condition (Claim~\ref{claim:equilibrium-condition}) for the case where designers choose $(T_1,\ldots,T_m)$, we get $\beta(C_i,\tilde p_i)=\beta(T_i,\tilde p_i)=\beta(T_j,\tilde p_j)=\beta(C_j,\tilde p_j)\ge\beta(T_\ell,\tilde p_\ell)=\beta(C_\ell,\tilde p_\ell)$ for any $i,j \in P$ and $\ell \notin P$. We therefore know that when designers choose $(C_1,\ldots,C_m)$, $(\tilde p_1,\ldots,\tilde p_m)$ is still a best response for any contestant when all the other contestants use $(\tilde p_1,\ldots,\tilde p_m)$, which means that $(\tilde p_1,\ldots,\tilde p_m)$ is a contestants' symmetric equilibrium. 

\subsection{Pareto efficiency of the equilibria: proof of Theorem~\ref{thm:pareto-optimality}}

Assume $\bm{T} = (T_1,\ldots,T_m)$ is the contestant-symmetric subgame-perfect equilibrium in Theorem~\ref{thm:max-rent-dissipation-NE} for $CCG(m,n,(R_i)_{i=1}^n,(\mathcal{S}_i)_{i=1}^n)$, and $\bm{C} = (C_1,\ldots,C_m)\in\mathcal{S}_1\times\cdots\times\mathcal{S}_m$ is any other strategy profile of the designers. We will show that $\bm{C}$ is not a Pareto improvement of $\bm{T}$ which proves the theorem. 
Denote by $(\tilde p_1,\ldots,\tilde p_m) = \bm p(\bm T)$ and $(p_1,\ldots,p_m) = \bm p(\bm C)$ the symmetric equilibria contestants play under $\bm T$ and $\bm C$, respectively.
If $p_i=\tilde p_i$ for any $i$, then as $T_i$ is the maximal rent dissipation contest of $\mathcal{S}_i$, by Claim~\ref{claim:mrd-inequality} we have
\begin{equation*}
    \beta(T_i,p_i) \le \beta(C_i,p_i)
\end{equation*}
As a result, for any designer $i$, 
\begin{align*}
    u_i(\bm{T}) 
    &= R_i[1-(1-\tilde p_i)^n]-n\tilde p_i\beta(T_i,\tilde p_i)\\
    &= R_i[1-(1-p_i)^n]-np_i\beta(T_i,p_i)\\
    &\ge R_i[1-(1-p_i)^n]-np_i\beta(C_i,p_i)
    = u_i(\bm{C}),
\end{align*}
showing that $\bm C$ is not a Pareto improvement of $\bm T$.

Otherwise, there exist $i,j$ such that $p_i>\tilde p_i$ and $p_j<\tilde p_j$. Note that this implies $p_i,\tilde p_j>0$, so by equilibrium condition (Claim~\ref{claim:equilibrium-condition}), we get
\begin{equation}\label{eq:PO-chain-1}
    \beta(C_i,p_i) \ge \beta(C_j,p_j),
\end{equation}
and
\begin{equation}\label{eq:PO-chain-2}
    \beta(T_j, \tilde p_j) \ge \beta(T_i, \tilde p_i).
\end{equation}
Since $p_j<\tilde p_j$ and $T_j$ is a monotonically decreasing utility contest, 
by Claim~\ref{claim:monotone-decreasing} we have
\begin{equation}\label{eq:PO-chain-3}
    \beta(T_j,\tilde p_j) < \beta(T_j,p_j),
\end{equation}
Moreover, as $T_j$ is a maximal rent dissipation contest in $\mathcal{S}_j$, we have
\begin{equation}\label{eq:PO-chain-4}
    \beta(T_j,p_j) \le \beta(C_j,p_j)
\end{equation}
by Claim~\ref{claim:mrd-inequality}. Combining these inequalities together, we get
\begin{equation*}
    \beta(T_i, \tilde p_i) \stackrel{\eqref{eq:PO-chain-2}}{\le} \beta(T_j, \tilde p_j) \stackrel{\eqref{eq:PO-chain-3}}{<} \beta(T_j, p_j) \stackrel{\eqref{eq:PO-chain-4}}{\le} \beta(C_j, p_j) \stackrel{\eqref{eq:PO-chain-1}}{\le} \beta(C_i, p_i).
\end{equation*}
Now we consider the utilities of designer $i$ in $\bm T$ and $\bm C$. We have
\begin{align*}
    u_i(\bm{T}) &= R_i[1-(1-\tilde p_i)^n]-n\tilde p_i\beta(T_i,\tilde p_i),\\
    u_i(\bm{C}) &= R_i[1-(1-p_i)^n]-np_i\beta(C_i,p_i).
\end{align*}
Similarly to \eqref{eq:f(p)}, we define $f(p)=R_i[1-(1-p)^n]-np\beta(T_i,\tilde p_i)$ and have
\begin{align*}
    f'(p) \le nR_i(1-p)^{n-1}-nR_i(1-\tilde p_i)^{n-1} < 0
\end{align*}
for $p>\tilde p_i$, which implies that $f(p)$ is a strictly decreasing function of $p$ when $p \ge \tilde p_i$. Therefore, as $p_i>\tilde p_i$, we have $f(p_i)<f(\tilde p_i)$. As a result,
\begin{align*}
    u_i(\bm{T})
    = f(\tilde p_i) >f(p_i) & = R_i[1-(1-p_i)^n]-np_i\beta(T_i,\tilde p_i)\\
    & \ge R_i[1-(1-p_i)^n]-np_i\beta(C_i,p_i) = u_i(\bm{C}), 
\end{align*}
which indicates that $\bm C$ cannot be a Pareto improvement of $\bm T$, concluding the proof. 

\section{Missing Proofs from Section~\ref{sec:welfare}}
\label{app:welfare_proofs}
\subsection{Proof of Theorem~\ref{thm:total-welfare-four-cases}}
For any $\bm C=(C_1,\ldots,C_m)$, let $\bm p(\bm C)=(p_1,\ldots,p_m)$.
By Eq.~\eqref{eq:WSC},
$$
    W_S(C_1,\ldots,C_m)=\sum_{i=1}^mR_i-\sum_{i=1}^mR_i(1-p_i)^n.$$
Then by H\"older's inequality,
\begin{align*}
    &\phantom{=} \left(\sum_{i=1}^mR_i(1-p_i)^n\right)^{\frac{1}{n}}
    \left(\sum_{i=1}^mR_i^{-\frac{1}{n-1}}\right)^{\frac{n-1}{n}}
    = \left(\sum_{i=1}^m\left(R_i^{\frac{1}{n}}(1-p_i)\right)^n\right)^{\frac{1}{n}}
    \left(\sum_{i=1}^m\left(R_i^{-\frac{1}{n}}\right)^{\frac{n}{n-1}}\right)^{\frac{n-1}{n}}\\
    &\ge \sum_{i=1}^m\left(R_i^{\frac{1}{n}}(1-p_i)\right)\left(R_i^{-\frac{1}{n}}\right)
    = \sum_{i=1}^m(1-p_i)
    = m-1.
\end{align*}
So
\begin{equation*}
    \sum_{i=1}^mR_i(1-p_i)^n
    \ge\left(\frac{m-1}{\left(\sum_{i=1}^mR_i^{-\frac{1}{n-1}}\right)^{\frac{n-1}{n}}}\right)^n
    =\frac{(m-1)^n}{\left(\sum_{i=1}^mR_i^{-\frac{1}{n-1}}\right)^{n-1}},
\end{equation*}
and
\begin{equation}
\label{eq:total-welfeare-bound}
    W_S(C_1,\ldots,C_m)=\sum_{i=1}^mR_i-\sum_{i=1}^mR_i(1-p_i)^n
    \le\sum_{i=1}^mR_i-\frac{(m-1)^n}{\left(\sum_{i=1}^mR_i^{-\frac{1}{n-1}}\right)^{n-1}}.
\end{equation}
We will show that $W_S(T_1,\ldots,T_m)$ is equal to the right-hand side in the following cases.

\vspace{5mm}

\noindent
{\bf  The case when $\mathcal{S}_i$ contains a full rent dissipation contest for every $i$. } In this case, since $T_i \in \MRD(\mathcal{S}_i)$, we have that $T_i$ is also a full rent dissipation contest. Let $\tilde p_i=p_i(T_1,\ldots,T_m)$. If $\tilde p_i=0$ for some $i$, then by Claim~\ref{claim:participation-probabilities-po}, $p_i(C_1,\ldots,C_m)=0$. Thus $W_S(\bm T) = W_S(\bm T_{-i})$ and $W_S(\bm C) = W_S(\bm C_{-i})$. We can thus assume that $\tilde{p}_i>0$ for every $i$.
%
%
Then by Claim \ref{claim:equilibrium-condition},
\begin{equation}\label{eq:equilibrium-condition-of-wta}
\beta(T_1,\tilde p_1)=\beta(T_2,\tilde p_2)=\cdots=\beta(T_m,\tilde p_m).
\end{equation}
As $T_i$ has full rent dissipation, we have $\beta(T_i,\tilde p_i)=R_i(1-\tilde p_i)^{n-1}$. 
Substitute this into \eqref{eq:equilibrium-condition-of-wta}, we get
\[ R_1^{\frac{1}{n-1}}(1-\tilde p_1) = R_2^{\frac{1}{n-1}}(1-\tilde p_2) = \cdots = R_m^{\frac{1}{n-1}}(1-\tilde p_m). \]
So $\forall i=1,\ldots,m$,
\begin{equation}
\label{eq:fraction}
\frac{1-\tilde p_j}{1-\tilde p_i}=\frac{R_j^{-\frac{1}{n-1}}}{R_i^{-\frac{1}{n-1}}}, \quad \forall j=1,\ldots,m.
\end{equation}
Note that $\sum_{j=1}^m (1-\tilde p_j) = m - 1$. Fixing any $i$ and summing \eqref{eq:fraction} for $j=1$ to $m$, we obtain
\[
\frac{m-1}{1-\tilde p_i}=\sum_{j=1}^m\frac{1-\tilde p_j}{1-\tilde p_i}=\frac{\sum_{j=1}^mR_j^{-\frac{1}{n-1}}}{R_i^{-\frac{1}{n-1}}}.
\]
So
\[
1-\tilde p_i=\frac{m-1}{\sum_{j=1}^mR_j^{-\frac{1}{n-1}}} \cdot R_i^{-\frac{1}{n-1}},
\]
and therefore,
\begin{equation*}
    W_S(T_1, \ldots, T_m)=\sum_{i=1}^mR_i-\sum_{i=1}^mR_i(1-\tilde p_i)^n=\sum_{i=1}^mR_i-\frac{(m-1)^n}{\left(\sum_{j=1}^mR_j^{-\frac{1}{n-1}}\right)^{n-1}},
\end{equation*}
which meets the bound given by \eqref{eq:total-welfeare-bound}. This completes the proof for this case.

\vspace{5mm}

\noindent
{\bf  The case of MRD-symmetric strategy space. } Note that by definition any two contests $T,T' \in \MRD(\mathcal{S}_i)$ must have the same $\bm{\cutil}$ vector (i.e., for any $k=1,\ldots,n$, $\cutil_{T}(k)=\cutil_{T'}(k)$). The following lemma therefore proves this case:


\begin{lemma}
In a $CCG(m, n, (R_i)_{i=1}^m, (\mathcal S_i)_{i=1}^m)$ with $R_1=\cdots = R_m=R$ (and $\mathcal S_i$'s can be different), for any strategy profile $\bm C' = (C_1', \ldots, C_m')$ where for every $i$, $C_i'\in\mathcal S_i$ is a MDU contest, and has the same $\bm{\cutil}$ vector (i.e., for any $k=1,\ldots,n$, and any $i,j\in\{1,\ldots,m\}$, $\cutil_{C_i'}(k)=\cutil_{C_j'}(k)$), $\bm C'$ maximizes $W_S$. 
\end{lemma}

\begin{proof}
When the rewards are the same, \eqref{eq:total-welfeare-bound} becomes
\begin{equation*}
    W_S(C_1,\ldots,C_m)=mR-R\sum_{i=1}^m(1-p_i)^n
    \le mR-R\frac{(m-1)^n}{m^{n-1}}
\end{equation*}
for any $(C_1,\ldots,C_m) \in \mathcal S_1 \times \cdots \times \mathcal S_m$.
On the other hand, since any two $C'_i$ and $C'_j$ have the same $\bm{\cutil}$ vector it follows that
\[
\beta\left(C_1',\frac{1}{m}\right)=\cdots=\beta\left(C_m',\frac{1}{m}\right).
\]
So $(p_1, \ldots, p_m) = \left(\frac{1}{m}, \ldots, \frac{1}{m}\right)$ is a symmetric equilibrium for the contestants under $(C_1', \ldots, C_m')$. Moreover, by Lemma~\ref{lem:MDU-contestant-unique}, it is the unique symmetric equilibrium, so $p_i(\bm C')=\frac{1}{m}$.  Therefore
\begin{equation*}
    W_S(C_1', \ldots, C_m')=\sum_{i=1}^mR_i-\sum_{i=1}^mR_i(1-p_i(\bm C'))^n=mR-R\frac{(m-1)^n}{m^{n-1}} \ge W_S(C_1,\ldots,C_m),
\end{equation*}
which completes our proof.
\end{proof}

\section{Additional Examples}\label{app:example}

\subsection{Non-MDU contests}\label{app:non-MDU}
Theorem~\ref{thm:max-rent-dissipation-NE} (or Corollary~\ref{cor:full-rent-dissipation-NE}) show that a sufficient condition for a set of contests to be in equilibrium in a contest competition game is that the contests have both MRD and MDU (or simply having full rent dissipation).
Example~\ref{ex:non-MDU} then shows that this condition is not necessary by giving an example of equilibria consisting of non-MDU contests; also, it shows that the MRD and MDU contests are not dominant when the sets $\mathcal S_i$ contain non-MDU contests.  But that example makes use of an unnatural Tullock contest with parameter $\tau_k=0$, where the reward is given to contestants for free.  Here we give another example using a more natural Tullock contest with $\tau_k=1$.  
\begin{example}\label{ex:non-monotone-tullock}
 Let $m=2, n=10$, $R_1 = R_2 = 1$, both $\mathcal S_1$ and $\mathcal S_2$ consist of two contests: the $\wta$ contest and a contest $C$ with $\cutil_C = (1, 0, 0, 0, 0, 0, 1/49, 1/64, 1/81, 1/100)$ (that is, choosing Tullock contest with $\tau_k=1$ when there are $7\le k\le 10$ contestants and $\tau_k=+\infty$ otherwise).  We claim that $(C, C)$ is a contestant-symmetric subgame-perfect equilibrium:  When designers choose $(C, C)$, by symmetry, contestants participate in either contest with equal probability $(0.5, 0.5)$.  By direct computation (e.g. using \eqref{eq:designer-utility}), the expected utility of each designers is
 \[  (0.9658, 0.9658
). \]
Now suppose designer 1 switches to $\wta$.  The probabilities $(p_1, p_2) = \bm p(\wta, C)$ in the contestants' symmetric mixed strategy Nash equilibrium must satisfy, according to Claim~\ref{claim:equilibrium-condition}, $\beta(\wta, p_1) = \beta(C, p_2)$ (assuming $p_1, p_2>0$). By numerical methods, we find that $(p_1, p_2) = (0.4125, 0.5875)$.  Then the expected utility of designers becomes
\[ (0.9607, 0.9509). \]
Since $0.9607 < 0.9658$, designer 1 will not switch to $\wta$.  By symmetry, designer 2 will not switch to $\wta$, either.  Hence, $(C, C)$ is an equilibrium.  Also, $\wta$ is not a dominant contest because it is not a best-response to $C$. 
\end{example}

Concerning dominant contests in contest competition games, the second part of Theorem~\ref{thm:max-rent-dissipation-NE} show that the MRD and MDU contests (e.g., the full rent dissipation contest $\wta$) are dominant if all sets $\mathcal S_i$ only contain MDU contests.   Example~\ref{ex:non-MDU} and Example~\ref{ex:non-monotone-tullock} show that $\wta$ is not dominant for designer $1$ when both sets $\mathcal S_1$ and $\mathcal S_2$ contain non-MDU contests in a two-designer game.  Here we give another example where $\wta$ is not dominant even if $\mathcal S_1$ only contains MDU contests (while $S_2$ still contains non-MDU contests).   

\begin{example}\label{ex:wta-not-dominant}
Consider two contest designers ($m=2$) with $R_1=R_2=1$ and ten contestants ($n=10$).
Designer 2 chooses contest $C$ which is a Tullock contest whose parameter $\tau_k$ depends on the number of contestants $k$: if $k\le 5$, she plays $\wta$; if $k\ge 6$, she chooses $\tau_k=1$.
Consider designer 1 that chooses from the set of all Tullock contests.
We show that $\wta$ is not a dominant contest for designer 1 by showing that the Tullock contest with $\tau=1.2$ is a better response for designer 1 than $\wta$. 

According to Lemma~\ref{lem:tullock}, the contest utility functions for contestants satisfy:
$$\cutil_{C}(k) = \begin{cases}
1 & k=1 \\
0 & 2 \leq k \leq 5 \\
\frac{1}{k^2} & k \geq 6
\end{cases} \quad 
\cutil_{\wta}(k) = \begin{cases}
1 & k=1 \\
0 & k \geq 2
\end{cases} \quad 
\cutil_{\tau=1.2}(k) = \begin{cases}
1 & k=1 \\
\frac{1}{k} - \frac{1.2(k-1)}{k^2} & 2 \leq k \leq 5 \\
0 & k\geq 6
\end{cases}
$$

Using the contestants' equilibria equations (Claim~\ref{claim:equilibrium-condition}), this induces participation equilibria probabilities of:
$$\hat{p}_1 \approx 0.366965,~~ p_1 \approx 0.519786, $$
where $(\hat p_1, \hat p_2) = \bm p(\wta, C)$ and $(p_1, p_2) = \bm p(\tau=1.2, C)$.

This translates to designer 1 utility of $\approx 0.929759$ choosing $\wta$, and $\approx 0.930121$ choosing $\tau=1.2$. 
\end{example}

\subsection{Proof of Example~\ref{ex:equilibria-asymmetric-rewards}}\label{sec:proof-example-asymmetric-reward}

\equilibriaAsymmetricRewards*

\begin{proof}
When for all $i=1, \ldots, m$, contest designer $i$ chooses $T_i=\wta$, we have that for any $j=2,\ldots,m$,
\[\beta\left(T_j,\frac{1}{m-1}\right)=R_j\left(1-\frac{1}{m-1}\right)^{n-1}=\left[\left(\frac{m-1}{m-2}\right)^{n-1}+1\right]\left(\frac{m-2}{m-1}\right)^{n-1}>1=R_1=\beta(T_1,0),\]
which implies that $\left(0,\frac{1}{m-1},\ldots,\frac{1}{m-1}\right)$ is a symmetric equilibrium for contestants in $(T_1,\ldots,T_m)$, i.e. $p_i(T_1,\ldots,T_m)=\begin{cases}\frac{1}{m-1}, & i=2,\ldots,m\\0, & i=1\end{cases}$.
Then the claim that any $(C_1, T_2, \ldots, T_m)$, with $T_j=\wta$ for $j=2,\ldots,m$ and $C_1$ be an arbitrary contest in $\mathcal S_1$, is a contestant-symmetric subgame-perfect equilibrium with $p_i(C_1, T_2, \ldots,T_m)=p_i(T_1,\ldots,T_m)=\begin{cases}\frac{1}{m-1}, & i=2,\ldots,m\\0, & i=1\end{cases}$ follows from Theorem~\ref{thm:uniqueness-NE}. 
\end{proof}

\subsection{Asymmetric participation equilibrium}
To demonstrate the importance of the symmetry aspect of the contestant-symmetric subgame-perfect equilibrium to our results, we consider an example where we allow contestants to play an asymmetric equilibrium in response to the designers' contest success functions; the conclusion that $\wta$ contests form an equilibrium now no longer holds. 
\begin{example}\label{ex:asymmetric}
Consider two contest designers ($m=2$) with reward $R_1=R_2=1$ and three contestants ($n=3$). The designers choose Tullock contests with parameter $\tau \in [0,\infty]$. We show that if the contestants may play an asymmetric equilibrium for their participation, $(\wta,\wta)$ may not be an equilibrium for the designers. We compare designer's 1 utility of choosing either $\tau = \wta$ or $\tau = 1.5$ in response to designer 2 choosing $\wta$. 

If designer 1 chooses $\wta$ and the contestants play the symmetric equilibria of $p_1 = p_2 = \frac{1}{2}$, the designers' expected utilities are the same. Moreover, notice that under any realization of the participation probabilities, there is some contest with at least 2 contestants and some contest with at most 1 contestant. By the full rent dissipation property of $\wta$, we conclude that under any realization the sum of utilities for the designers is 1, and so each has an expected utility of $\frac{1}{2}$. 

If designer 1 chooses $\tau= 1.5$, the contestants may play an asymmetric equilibrium where contestants 1,2 choose contest 1 and contestant 3 chooses contest 2. This is an equilibrium because contestant 3 gets the full reward with no effort, and for contestant 1 (w.l.o.g., the same argument applies to contestant 2), given the other contestants' choices and that changing her choice to contest 2 leads to an expected utility of zero, the minimal possible, 
it is a best response to stay in contest 1. 

Under this asymmetric participation equilibrium, the utility for contest 1's designer is $0.75$ (according to Lemma~\ref{lem:tullock}, $u_1 = R_1 - k \gamma_{\tau=1.5}(k) = 1 - 2 \cdot 0.125 = 0.75$), higher than its utility by setting $\wta$. We thus conclude that there is a better response than $\wta$ for contest designer 1 to the $\wta$ set by contest designer 2, and therefore $(\wta,\wta)$ is not an equilibrium. This stands in contrast with the results of Theorem~\ref{thm:max-rent-dissipation-NE} which assumes symmetric equilibrium. 
\end{example}

\subsection{Risk averse contestants}
\label{app:risk-averse}

Consider an example CCG with $m=2, n=2, R_1=R_2=1$ and  $\mathcal S_i$ for both contest designers is the set of Tullock contests with $\tau \in [0,2]$. With risk-neutral contestants we showed that $(\tau=2,\tau=2)$ is the unique contestant-symmetric subgame-perfect equilibrium (Corollary~\ref{cor:full-rent-dissipation-NE} and Corollary~\ref{cor:uniqueness-NE-symmetric-reward}).
We show via an example that with risk-averse contestants $(\tau=2,\tau=2)$ is no longer a contestant-symmetric subgame-perfect equilibrium.

\begin{definition}
A risk-averse contestant in a contest has a twice differentiable, strictly increasing in $[0,1]$, concave utility function $a:\mathcal{R} \rightarrow \mathcal{R}$, with $a(0) = 0, a(1) = 1$. 
\end{definition}

In our example we use the utility function $a(x) = 1 - (1-x)^4$. \cite{segev_schweinzer_2012} show that, in a single Tullock contest with parameter $\tau \in [0,2]$, there exists a pure strategy symmetric Nash equilibrium (see also the proof of Lemma~\ref{lem:tullock}). Since it is pure, the efforts exerted in this equilibrium are the same whether contestants are risk-neutral or risk-averse.
A contestant's utility in the competition game is thus
(following Eq.~\eqref{eq:contestant-utility-beta}):
\begin{equation}\label{eq:contestant-utility-beta-averse}
 \beta_{\AVERSE}(C_i, p_i) = \Ex[k\sim \bin(n-1, p_i)]{a(\cutil_{C_i}(k+1))} = \sum\limits_{k=0}^{n-1} \binom{n-1}{k} p_i^k (1-p_i)^{n-1-k} a(\cutil_{C_i}(k+1)).
\end{equation}

In particular, for $\tau=2$,
 $$\beta_{\AVERSE}(\tau=2, p_i) = \Ex[k\sim \bin(n-1, p_i)]{a(\cutil_{\tau=2}(k+1))} = (1-p_i)^{n-1} a(\cutil_{\tau=2}(1)) = (1-p_i)^{n-1}.$$

Claim~\ref{claim:equilibrium-condition} (equilibrium condition) continues to hold, 
and the definition of designer utilities remains the same as they remain risk-neutral (Eq.~\eqref{eq:designer-utility-definition}). 
%
%
%
%
We show that with risk-averse contestants $\tau = 1$ is a better response than $\tau=2$ to $\tau=2$. By symmetry, if both designers set $\tau = 2$, we have $p_1 = p_2 = \frac{1}{2}$ and designer 1's utility is $p_1^2 \cdot (1 - 2\cutil_{\tau=2}(2)) = \frac{1}{4}$ (the formula for $\cutil_{\tau}(k)$ is given in Lemma~\ref{lem:tullock}). If designer 1 sets $\tau = 1$, by the equilibrium condition for risk-averse participants, we have
    $$(1-p_1) a(1) + p_1 a(\cutil_{\tau=1}(2)) = \beta_{\AVERSE}(\tau = 1,p_1) = \beta_{\AVERSE}(\tau=2, p_2) = 1-p_2  = p_1,$$
where $\cutil_{\tau=1}(2) = \frac{1}{4}$, 
which yields $p_1 = \frac{256}{337}$, and the utility for designer 1 is
$p_1^2 (1 - 2\cutil_{\tau=1}(2)) \approx 0.288529$.
%
%
This establishes that $(\tau=2, \tau=2)$ is not an equilibrium. It can be verified that a symmetric equilibrium exists for this setting with $\tau = \frac{2}{3}$.


\bibliographystyle{apalike}
\bibliography{bibs}

\end{document}

%% file: notation.tex
\newcommand{\AutoAdjust}[3]{{ \mathchoice{ \left #1 #2  \right #3}{#1 #2 #3}{#1 #2 #3}{#1 #2 #3} }}
\newcommand{\Xcomment}[1]{{}}

\newcommand{\InBrackets}[1]{\AutoAdjust{[}{#1}{]}}
\newcommand{\Ex}[2][]{\operatorname{\mathbbm E}_{#1}\InBrackets{#2}}

\renewcommand{\vec}{\bm}

\newtheorem{theorem}{Theorem}[section]

\newtheorem{lemma}[theorem]{Lemma}
\newtheorem{definition}[theorem]{Definition}
\newtheorem{claim}[theorem]{Claim}
\newtheorem{corollary}[theorem]{Corollary}

\newtheorem{remark}[theorem]{Remark}
\newtheorem{example}[theorem]{Example}

\DeclareMathOperator*{\argmax}{arg\,max}

\newcommand{\noaccents}[1]{#1}
\newcommand{\newagentvar}[3][\noaccents]{%
\expandafter\newcommand\expandafter{\csname #2\endcsname}{#1{#3}}%
\expandafter\newcommand\expandafter{\csname #2s\endcsname}{#1{\boldsymbol{#3}}}%
\expandafter\newcommand\expandafter{\csname #2smi\endcsname}[1][i]{#1{\boldsymbol{#3}}_{-##1}}%
\expandafter\newcommand\expandafter{\csname #2i\endcsname}[1][i]{#1{#3}_{##1}}%
\expandafter\newcommand\expandafter{\csname #2ith\endcsname}[1][i]{#1{#3}_{(##1)}}%
}

\newcommand{\newvecagentvar}[3][\noaccents]{%
\expandafter\newcommand\expandafter{\csname #2\endcsname}{#1{\boldsymbol{#3}}}%
\expandafter\newcommand\expandafter{\csname #2s\endcsname}{#1{\boldsymbol{#3}}}%
\expandafter\newcommand\expandafter{\csname #2smi\endcsname}[1][i]{#1{\boldsymbol{#3}}_{-##1}}%
\expandafter\newcommand\expandafter{\csname #2i\endcsname}[1][i]{#1{\boldsymbol{#3}}_{##1}}%
\expandafter\newcommand\expandafter{\csname #2ith\endcsname}[1][i]{#1{#3}_{(##1)}}%
}

\newagentvar{val}{v}
\newagentvar{bid}{b}
\newagentvar{dist}{F}
\newagentvar{alloc}{x}
\newagentvar{util}{u}
\newagentvar{pay}{p}
\newagentvar{ability}{a}
\newagentvar{effort}{e}

\newcommand{\redcom}[1]{\textcolor{red}{[#1]}}

\DeclareMathOperator{\reals}{{\mathbb R}}

\newcommand{\bin}{\mathrm{Bin}}
\newcommand{\cutil}{\gamma}
\newcommand{\mrd}{\mathrm{MRD}}
\newcommand{\MRD}{\mathrm{MRD}}
\newcommand{\wta}{\mathrm{APA}}
\newcommand{\WTA}{\mathrm{APA}}
\newcommand{\AVERSE}{\mathrm{averse}}

\newcommand{\Supp}{\mathrm{Supp}}

\newcommand{\yotam}[1]{{\color{red}{Yotam: #1}}}